\documentclass[a4paper,UKenglish,cleveref,autoref,thm-restate]{lipics-v2021}

\hideLIPIcs  

\title{A Complete Diagrammatic Calculus for Conditional Gaussian Mixtures}

\author{Mateo Torres-Ruiz}{University College London, United Kingdom}{m.torresruiz@cs.ucl.ac.uk}{}{}

\author{Robin Piedeleu}{University College London, United Kingdom}{}{0000-0002-3945-2704}{}

\author{Alexandra Silva}{Cornell University, United States}{}{0000-0001-5014-9784}{}

\author{Fabio Zanasi}{University College London, United Kingdom}{}{0000-0001-6457-1345}{}

\authorrunning{M. Torres-Ruiz, R. Piedeleu, A. Silva, F. Zanasi} 

\Copyright{Mateo Torres-Ruiz, Robin Piedeleu, Alexandra Silva, Fabio Zanasi} 



\category{} 

\relatedversion{} 

\funding{This work was supported by \textsc{aria}'s Safeguarded AI  programme.}

\nolinenumbers 

\usepackage[utf8]{inputenc}
\usepackage{amsmath}
\usepackage{amssymb}
\usepackage{float}
\usepackage{thm-restate}

\usepackage[colorinlistoftodos,prependcaption,textsize=tiny]{todonotes}
\usepackage{soul}

\usepackage{graphicx}
\usepackage[all]{xy}
\usepackage{stmaryrd}
\usepackage{mathtools}
\usepackage[braket]{qcircuit}

\usepackage{calc}
\usepackage{mdframed}
\usepackage{bussproofs}
\usepackage{quiver}
\usepackage{tikz}
\usetikzlibrary{cd,babel}
\usetikzlibrary{braids}
\usetikzlibrary{positioning, shapes}

\pgfdeclarelayer{wires}
\pgfdeclarelayer{nodes}
\pgfsetlayers{wires, main, nodes}

\tikzstyle{none}=[inner sep=0mm]
\tikzstyle{not}=[draw=blue, triangle, rotate=90, draw=black, fill=white, inner sep=0pt, minimum size=4.5pt]
\tikzstyle{cflip}=[circle, draw={rgb,255: red,128; green,128; blue,128}, fill={rgb,255: red,128; green,128; blue,128}, inner sep=0pt, minimum size=4pt]
\tikzstyle{or}=[fill=white, draw=black, or gate, scale=.4, anchor=center]
\tikzstyle{not}=[fill=white, draw=black, not gate, scale=.35]
\tikzstyle{if}=[trapezium, draw=black, fill=white, minimum width=6pt, minimum height=8pt, rotate=270]
\tikzstyle{plain}=[inner sep=0pt]
\tikzstyle{red}=[circle, draw=realtype, fill=realtype, inner sep=0pt, minimum size=4pt, scale=0.75]
\tikzstyle{blue}=[circle, draw=booltype, fill=booltype, inner sep=0pt, minimum size=4pt, scale=0.75]
\tikzstyle{black-faded}=[circle, draw=light-gray, fill=light-gray, inner sep=0pt, minimum size=4pt]
\tikzstyle{white}=[circle, draw=realtype, fill=white, inner sep=0pt, minimum size=4.5pt]
\tikzstyle{white-faded}=[circle, draw=light-gray, fill=white, inner sep=0pt, minimum size=4.5pt]
\tikzstyle{reg}=[draw=red, fill=red, rounded rectangle, rounded rectangle left arc=none, minimum height=1.2em, minimum width=1.4em, text=white, node font={\scriptsize}]
\tikzstyle{effect}=[draw, fill=white, rectangle, minimum height=1.2em, minimum width=1em]
\tikzstyle{coreg}=[draw, fill=white, rounded rectangle, rounded rectangle right arc=none, minimum height=1.2em, minimum width=1.4em, node font={\scriptsize}]
\tikzstyle{state}=[draw, fill=white, rectangle, minimum height=1.2em, minimum width=1em]
\tikzstyle{big box}=[draw, fill=white, rectangle, rounded corners, minimum height=3em, minimum width=2.4em, node font={\scriptsize}]
\tikzstyle{small box}=[draw, fill=white, rectangle, rounded corners, minimum height=1.2em, minimum width=1.4em, node font={\scriptsize}]
\tikzstyle{nn}=[fill=realtype, draw=realtype, regular polygon, regular polygon sides=3, minimum width=0.3cm, shape border rotate=90, inner sep=0pt, scale=0.75]

\tikzstyle{thick}=[-, line width=1pt]
\tikzstyle{dashed edge}=[-, dashed]
\tikzstyle{bool edge}=[-, draw=booltype]
\tikzstyle{thick bool edge}=[-, draw=booltype, line width=1pt]
\tikzstyle{real edge}=[-, draw=realtype]
\tikzstyle{thick real edge}=[-, draw=realtype, line width=1pt]

\def\distto{\mathrel{\mkern3mu  \vcenter{\hbox{$\scriptscriptstyle+$}}%
\mkern-12mu{\to}}}
\newcommand{\Normal}[3]{\mathcal{N}_{#1}\left({#2},{#3}\right)}
\newcommand{\R}{\mathbb{R}}
\newcommand{\Realobj}{\mathsf{R}}
\newcommand{\N}{\mathbb{N}}
\newcommand{\B}{\mathbb{B}}
\newcommand{\Boolobj}{\mathsf{B}}

\newcommand{\Borel}[1]{\mathcal{B}\left({#1}\right)}
\newcommand*\diff{\mathop{}\!\mathrm{d}}

\newcommand\restr[2]{{
\left.\kern-\nulldelimiterspace
#1
\vphantom{\big|}
\right|_{#2}
}}
\newcommand\textax[1]{\textbf{\textsf{#1}}}
\newcommand\textsub[2]{$#1|_{#2}$}

\newcommand{\Bool}{\mathbb{B}}
\newcommand{\Real}{\mathbb{R}}

\newcommand{\poi}{\,;\,}
\newcommand{\id}{\mathsf{id}}

\newcommand{\Boolid}{
\begin{tikzpicture}
\begin{pgfonlayer}{nodelayer}
\node [style=none] (28) at (0.5, 0) {};
\node [style=none] (33) at (-1.25, 0) {};
\end{pgfonlayer}
\begin{pgfonlayer}{edgelayer}
\draw [style=bool edge] (33.center) to (28.center);
\end{pgfonlayer}
\end{tikzpicture}}
\newcommand{\Realid}{
\begin{tikzpicture}
\begin{pgfonlayer}{nodelayer}
\node [style=none] (28) at (0.5, 0) {};
\node [style=none] (33) at (-1.25, 0) {};
\end{pgfonlayer}
\begin{pgfonlayer}{edgelayer}
\draw [style=real edge] (33.center) to (28.center);
\end{pgfonlayer}
\end{tikzpicture}
}

\newcommand{\idone}{
\tikz \draw (0, 0) -- (2, 0);
}
\newcommand{\idonebool}{
\tikz \draw [style=bool edge] (0, 0) -- (2, 0);
}

\newcommand\idzero{
\InputIfFileExists{empty-diag.tikz}{}{\input{./tikz/empty-diag.tikz}}
} 
\newcommand{\sym}{
\tikz {
\draw (0,  0.4) .. controls (0.5,  0.4) and (0.5, -0.4) .. (1, -0.4);
\draw (0, -0.4) .. controls (0.5, -0.4) and (0.5,  0.4) .. (1,  0.4);
}
}
\newcommand{\symp}[2]{
\tikz {
\draw[#1] (0,  0.4) .. controls (0.5,  0.4) and (0.5, -0.4) .. (1, -0.4);
\draw[#2] (0, -0.4) .. controls (0.5, -0.4) and (0.5,  0.4) .. (1,  0.4);
}
}

\newcommand{\circuit}[3]{
\begin{tikzpicture}
\begin{pgfonlayer}{background}
\node [style=circuit] (0) at (0, 0) {$#1$};
\node [style=none] at (1.5, 0.5) {\scriptsize $#3$};
\node [style=none] at (-1.5, 0.5) {\scriptsize $#2$};
\end{pgfonlayer}
\begin{pgfonlayer}{nodelayer}
\node [style=none] (1) at (1.75, 0.05) {};
\node [style=none] (2) at (-1.75, 0.05) {};
\node [style=none] (3) at (1.5, 0) {};
\node [style=none] (4) at (-1.5, 0) {};
\end{pgfonlayer}
\begin{pgfonlayer}{edgelayer}
\draw  (2) to (0);
\draw  (0) to (1);
\end{pgfonlayer}
\end{tikzpicture}
}

\newcommand{\mixcircuit}[5]{
\begin{tikzpicture}
\begin{pgfonlayer}{background}
\node [style=none] (1) at (0,0.4) {};
\node [style=none] (2) at (-1.75,0.4) {};
\node [style=none] (3) at (1.75, 0.4) {};
\node [style=none] (4) at (0,-0.4) {};
\node [style=none] (5) at (-1.75,-0.4) {};
\node [style=none] (6) at (1.75, -0.4) {};
\end{pgfonlayer}
\begin{pgfonlayer}{nodelayer}
\node [style=circuit] (0) at (0, 0) {$#1$};
\node [style=none] at (2.05, 0.45) {\scriptsize $#3$};
\node [style=none] at (-2.05, 0.45) {\scriptsize $#2$};
\node [style=none] at (2.1, -0.45) {\scriptsize $#5$};
\node [style=none] at (-2.1, -0.45) {\scriptsize $#4$};
\end{pgfonlayer}
\begin{pgfonlayer}{edgelayer}
\draw [style=thick bool edge] (2) to (1);
\draw [style=thick bool edge] (1) to (3);
\draw [style=thick real edge] (5) to (4);
\draw [style=thick real edge] (4) to (6);
\end{pgfonlayer}
\end{tikzpicture}
}

\newcommand{\thickboolcircuit}[3]{
\begin{tikzpicture}
\begin{pgfonlayer}{background}
\node [style=bool circuit] (0) at (0, 0) {$#1$};
\node [style=none] at (1.5, 0.5) {\scriptsize $#3$};
\node [style=none] at (-1.5, 0.5) {\scriptsize $#2$};
\end{pgfonlayer}
\begin{pgfonlayer}{nodelayer}
\node [style=none] (1) at (1.75, 0.05) {};
\node [style=none] (2) at (-1.75, 0.05) {};
\node [style=none] (3) at (1.5, 0) {};
\node [style=none] (4) at (-1.5, 0) {};
\end{pgfonlayer}
\begin{pgfonlayer}{edgelayer}
\draw [style=-, line width=1pt] (2) to (0);
\draw [style=-, line width=1pt] (0) to (1);
\end{pgfonlayer}
\end{tikzpicture}
}

\newcommand{\ifmix}{
\begin{tikzpicture}
\begin{pgfonlayer}{nodelayer}
\node [style=none] (91) at (-1.5, 0.5) {};
\node [style=if-mix] (93) at (0, 0) {};
\node [style=none] (96) at (-1.5, 0) {};
\node [style=none] (104) at (-0.25, 0) {};
\node [style=none] (105) at (-0.25, 0.5) {};
\node [style=none] (110) at (1.5, 0) {};
\node [style=none] (111) at (-1.5, -0.5) {};
\node [style=none] (112) at (-0.25, -0.5) {};
\end{pgfonlayer}
\begin{pgfonlayer}{edgelayer}
\draw [style=bool edge] (91.center) to (105.center);
\draw [style=real edge] (96.center) to (104.center);
\draw [style=real edge] (111.center) to (112.center);
\draw [style=real edge] (93) to (110.center);
\end{pgfonlayer}
\end{tikzpicture}
}

\newcommand{\MixCirc}{\mathsf{MixCirc}}
\newcommand{\MixGauss}{\mathsf{MixGauss}}
\newcommand{\CGMTh}{\mathsf{CGM}}
\newcommand{\CGMeq}{\myeq{CGM}}
\newcommand{\Stoch}{\mathsf{Stoch}}

\newcommand{\sem}[1]{\left\llbracket{#1}\right\rrbracket}
\newcommand{\Gauss}{\mathsf{Gauss}}
\newcommand{\GaussCirc}{\mathsf{GaussCirc}}
\newcommand{\CausCirc}{\mathsf{CausCirc}}
\newcommand{\FinStoch}{\mathsf{FinStoch}}

\newcommand{\CGtype}[4]{\Boolobj^{#1}\Realobj^{#2}\to \Boolobj^{#3}\Realobj^{#4}}
\newcommand{\CGterm}[5]{#1\colon\CGtype{#2}{#3}{#4}{#5}}
\newcommand{\GaussCircterm}[3]{#1\colon \Realobj^{#2}\to \Realobj^{#3}}

\newcommand{\Dirac}[2]{\delta\left({#1}\,\middle\vert\,{#2}\right)}

\usepackage{tikz}
\usetikzlibrary{arrows,snakes,automata,backgrounds, positioning,calc,cd}
\usetikzlibrary{shapes.geometric,shapes.misc,matrix,arrows.meta,decorations.markings}
\usetikzlibrary{circuits.logic.US, patterns, shapes}
\definecolor{comb}{HTML}{ddf1fb}
\definecolor{seq}{HTML}{6ae4a3}
\definecolor{neutral}{HTML}{f8f8f2}
\definecolor{unit}{HTML}{b9b9b9}

\newcommand{\dotsize}{3pt}
\tikzset{x=0.9em, y=2ex, baseline=-0.5ex}
\tikzset{ihbase/.style={inner sep=0,circle,draw,fill=lightgray,minimum size={\dotsize}, node contents={}}}
\tikzset{ihblack/.style={ihbase,fill=black}}
\tikzset{ihblue/.style={ihbase,fill=blue}}
\tikzset{ihwhite/.style={ihbase,fill=white}}
\tikzset{mat/.style={draw,fill=white,rectangle,node font=\scriptsize}}
\tikzset{ha/.style={mat,rounded rectangle,rounded rectangle left arc=none}}
\tikzset{haop/.style={mat,rounded rectangle,rounded rectangle right arc=none}}
\tikzset{blackha/.style={mat,rounded rectangle,rounded rectangle left arc=none,font=\color{white},fill=black}}
\tikzset{blackhaop/.style={mat,rounded rectangle,rounded rectangle right arc=none,font=\color{white},fill=black}}
\tikzset{anti/.style={inner sep=0,isosceles triangle,fill=black,draw=black, minimum width=0.75em, node contents={}}}
\tikzset{antiop/.style={anti,shape border rotate=180}}
\tikzset{antisq/.style={inner sep=0,rectangle,fill=black, minimum height=1em, minimum width=0.6em, node contents={}}}
\tikzset{count/.style={above,inner ysep=0.15em,font=\scriptsize}}
\tikzset{axiom/.style={above,font=\small}}
\tikzset{dir/.style={-Latex}}
\tikzset{st/.style={decoration={markings,
mark={at position 0.5 with {\draw (0, 2pt) to (0, -2pt);}}},
postaction=decorate}}
\tikzstyle{none}=[inner sep=0mm]
\tikzstyle{flip}=[draw={booltype}, fill={booltype},
rounded rectangle, rounded rectangle right arc=none, text=white, node font={\scriptsize}, inner sep =0.2em, scale=0.8]
\tikzstyle{cflip}=[circle, draw={rgb,255: red,128; green,128; blue,128}, fill={rgb,255: red,128; green,128; blue,128}, inner sep=0pt, minimum size=4pt]
\tikzstyle{scalar}=[draw={red}, fill={red},
rounded rectangle, rounded rectangle left arc=none, text=white, node font={\scriptsize}, text={white}, inner sep =0.2em]
\tikzstyle{and}=[fill={booltype}, draw={booltype}, shape=and gate US, scale=.4]
\tikzstyle{or}=[fill=white, draw={booltype}, or gate US, scale=.4, anchor=center]
\tikzstyle{not}=[fill={booltype}, draw={booltype}, not gate US, scale=.4]
\tikzstyle{xor}=[fill={booltype}, draw=black, xor gate US, scale=.6]
\tikzstyle{if-mix}=[trapezium, draw={booltype}, fill={mixedtype}, minimum width=6pt, minimum height=8pt, rotate=270]
\tikzstyle{plain}=[inner sep=0pt]
\tikzstyle{black}=[circle, draw=black, fill=black, inner sep=0pt, minimum size={\dotsize}]
\tikzstyle{black-faded}=[circle, draw=light-gray, fill=light-gray, inner sep=0pt, minimum size=4pt]
\tikzstyle{white}=[circle, draw=black, fill=white, inner sep=0pt, minimum size={\dotsize}]
\tikzstyle{white-faded}=[circle, draw=light-gray, fill=white, inner sep=0pt, minimum size=4.5pt]
\tikzstyle{reg}=[draw={realtype}, fill={realtype}, rounded rectangle, rounded rectangle left arc=none, minimum height=1.2em, minimum width=1.4em, text=white, node font={\scriptsize},scale=.8]
\tikzstyle{effect}=[draw, fill=white, rounded rectangle, rounded rectangle left arc=none, minimum height=1.2em, minimum width=1.4em]
\tikzstyle{coreg}=[draw, fill=white, rounded rectangle, rounded rectangle right arc=none, minimum height=1.2em, minimum width=1.4em, node font={\scriptsize}]
\tikzstyle{box}=[shape=rectangle, text height=1.5ex, text depth=0.25ex, yshift=0.2mm, fill=white, draw=black, minimum height=3mm, minimum width=5mm, font={\small}]
\tikzstyle{basic box}=[shape=rectangle, text height=1.5ex, text depth=0.25ex, yshift=0.2mm, fill={probcircuitcolor}, draw={probcircuitcolor}, minimum height=3mm, minimum width=5mm, text={probcircuitcolorop}, inner sep=4pt]
\tikzstyle{small box}=[draw, fill={probcircuitcolor}, rectangle, minimum height=1.2em, minimum width=1.4em, node font={\scriptsize}, text={probcircuitcolorop}]

\definecolor{realtype}{rgb}{0,0,0}
\definecolor{booltype}{rgb}{0.5,0.5,0.5}
\definecolor{mixedtype}{rgb}{1,1,1}

\definecolor{probcircuitcolor}{rgb}{0.85,0.85,0.85}
\definecolor{probcircuitcolorop}{rgb}{0,0,0}
\definecolor{boolcircuitcolor}{rgb}{1.0,1.0,1.0}
\definecolor{boolcircuitcolorop}{rgb}{0.0,0.0,0.0}
\definecolor{mixedcircuitcolor}{rgb}{0.6,0.6,0.6}
\tikzstyle{circuit}=[shape=rectangle, text height=1.5ex, text depth=0.25ex, yshift=0.2mm, fill=white, draw=black, minimum height=3mm, minimum width=5mm, text=black, inner sep=4pt]
\tikzstyle{prob circuit}=[shape=rectangle, text height=1.5ex, text depth=0.25ex, yshift=0.2mm, fill={realtype}, draw={realtype}, minimum height=3mm, minimum width=5mm, text=white, inner sep=4pt, font=\scriptsize]
\tikzstyle{small prob circuit}=[draw={realtype}, fill={realtype}, rectangle, minimum height=1.2em, minimum width=1.4em, node font={\scriptsize}, text=white]
\tikzstyle{tall prob circuit}=[draw, shape=rectangle, text height=1.5ex, text depth=0.25ex, yshift=0.2mm, fill={realtype}, minimum height=8mm, minimum width=5mm, text=white, font={\scriptsize}]
\tikzstyle{bool circuit}=[shape=rectangle, text height=1.5ex, text depth=0.25ex, yshift=0.2mm, fill={booltype}, draw={booltype}, minimum height=3mm, minimum width=5mm, font={\scriptsize}, text=white]
\tikzstyle{small bool circuit}=[draw={booltype}, fill={booltype}, rectangle, minimum height=1.2em, minimum width=1.4em, node font={\scriptsize}, text=white]
\tikzstyle{tall bool circuit}=[rectangle, draw={booltype}, fill={booltype}, minimum height=8mm, minimum width=5mm ,text=white, font={\scriptsize}]
\tikzstyle{taller bool circuit}=[rectangle, draw={booltype}, fill={booltype}, minimum height=16mm, minimum width=6mm ,text=white, font={\scriptsize}]
\tikzstyle{small mixed circuit}=[draw={booltype}, fill={mixedtype}, rectangle, minimum height=1.2em, minimum width=1.4em, text=black, font={\scriptsize}]
\tikzstyle{tall mixed circuit}=[rectangle, draw={booltype}, fill={mixedtype}, minimum height=8mm, minimum width=5mm ,text=black, font={\scriptsize}]
\tikzstyle{taller mixed circuit}=[rectangle, draw={booltype}, fill={mixedtype}, minimum height=12mm, minimum width=6mm ,text=black, font={\scriptsize}]
\tikzstyle{real type} = [black]
\tikzstyle{bool type} = [mixedcircuitcolor]
\tikzstyle{bold real type} = [black, line width=.75pt]
\tikzstyle{bold bool type} = [mixedcircuitcolor, line width=.85pt]

\pgfdeclarelayer{edgelayer}
\pgfdeclarelayer{nodelayer}
\pgfsetlayers{background,edgelayer,nodelayer,main}

\newcommand{\genericcounit}[2]{
\tikz \draw[#2] (0, 0) -- (1, 0) node[ihbase,#1, solid];
}
\newcommand{\genericcomult}[2]{
\begin{tikzpicture}
\node at (1, 0) [ihbase,solid,name=copy,#1];
\draw[#2] (copy) .. controls (1.25, 0.5) .. (2, 0.5);
\draw[#2] (0, 0) -- (copy);
\draw[#2] (copy) .. controls (1.25, -0.5) .. (2, -0.5);
\end{tikzpicture}
}

\newcommand{\genericunit}[2]{
\tikz \draw[#2] (0, 0) node[ihbase,fill=#1,#2, solid] -- (1, 0);
}

\newcommand{\genericmult}[2]{
\tikz {
\node at (1, 0) [ihbase,solid,fill=#1,name=copy,#2];
\draw[#2] (0,  0.5) .. controls (0.75,  0.5) .. (copy);
\draw[#2] (0, -0.5) .. controls (0.75, -0.5) .. (copy);
\draw[#2] (copy) -- (2, 0);
}
}

\newcommand{\BoolCounit}[1][]{\genericcounit{booltype}{style=#1,draw={booltype}}}
\newcommand{\BoolComult}[1][]{\genericcomult{booltype}{style=#1,draw={booltype}}}

\newcommand{\RealCounit}[1][]{\genericcounit{realtype}{style=#1,draw={realtype}}}

\newcommand{\RealComult}[1][]{\genericcomult{realtype}{style=#1,draw={realtype}}}

\newcommand{\RealWUnit}[1][]{\genericunit{white}{style=#1,draw={realtype}}}
\newcommand{\RealMult}[1][]{\genericmult{white}{style=#1,draw={realtype}}}

\newcommand{\Scalar}[1]{
\begin{tikzpicture}
\begin{pgfonlayer}{nodelayer}
\node [style=reg] (93) at (0, 0) {$#1$};
\node [style=none] (96) at (-1.25, 0) {};
\node [style=none] (103) at (1.25, 0) {};
\end{pgfonlayer}
\begin{pgfonlayer}{edgelayer}
\draw [color={realtype}] (96.center) to (93);
\draw [color={realtype}] (93) to (103.center);
\end{pgfonlayer}
\end{tikzpicture}   
}

\newcommand{\Flip}[1]{
\begin{tikzpicture}
\begin{pgfonlayer}{nodelayer}
\node [style=flip] (93) at (0, 0) {$#1$};
\node [style=none] (103) at (1.25, 0) {};
\end{pgfonlayer}
\begin{pgfonlayer}{edgelayer}
\draw [color={booltype}] (93) to (103.center);
\end{pgfonlayer}
\end{tikzpicture}
}

\newcommand{\Foot}{
\begin{tikzpicture}
\begin{pgfonlayer}{nodelayer}
\node [style=none] (0) at (0, 0) {};
\node [style=none] (1) at (1.15, 0) {};
\node [style=none] (2) at (0, 0.25) {};
\node [style=none] (3) at (0, -0.25) {};
\end{pgfonlayer}
\begin{pgfonlayer}{edgelayer}
\draw [color={realtype}] (0.center) to (1.center);
\draw [color={realtype}] (2.center) to (3.center);
\end{pgfonlayer}
\end{tikzpicture}
}

\newcommand{\Andgate}{
\begin{tikzpicture}
\begin{pgfonlayer}{nodelayer}
\node [style=none] (91) at (-1.5, 0.5) {};
\node [style=and gate US, draw={booltype}, fill={booltype}, scale=.4] (93) at (0, 0) {};
\node [style=none] (96) at (-1.5, -0.5) {};
\node [style=none] (103) at (1.25, 0) {};
\end{pgfonlayer}
\begin{pgfonlayer}{edgelayer}
\draw [color={booltype}, in=0, out=150] (93) to (91.center);
\draw [color={booltype}, in=-150, out=0, looseness=0.75] (96.center) to (93);
\draw [color={booltype}] (93) to (103.center);
\end{pgfonlayer}
\end{tikzpicture}
}

\newcommand{\Notgate}[1][]{
\begin{tikzpicture}
\begin{pgfonlayer}{nodelayer}
\node [style=not gate US, draw={booltype}, fill={booltype}, scale=.4] (93) at (0, 0) {};
\node [style=none] (96) at (-1.25, 0) {};
\node [style=none] (103) at (1.25, 0) {};
\end{pgfonlayer}
\begin{pgfonlayer}{edgelayer}
\draw [style=#1, color={booltype}] (96.center) to (93);
\draw [style=#1, color={booltype}] (93) to (103.center);
\end{pgfonlayer}
\end{tikzpicture}
}

\newcommand{\NN}{
\genericunit{{realtype}, regular polygon, regular polygon sides=3, minimum width=0.175cm, shape border rotate=90, inner sep=0pt, draw={realtype}, fill={realtype}}{draw={realtype}}}

\newcommand{\myeq}[1]{\mathrel{\overset{\makebox[0pt]{\mbox{\normalfont\tiny\sffamily #1}}}{=}}}

\newcommand{\eqSMC}{\equiv}

\begin{document}

\maketitle

\begin{abstract}
  We extend the synthetic theories of discrete and Gaussian categorical probability by introducing a diagrammatic calculus for reasoning about hybrid probabilistic models in which continuous random variables, conditioned on discrete ones, follow a multivariate Gaussian distribution. This setting includes important classes of models such as Gaussian mixture models, where each Gaussian component is selected according to a discrete variable. We develop a string diagrammatic syntax for expressing and combining these models, give it a compositional semantics, and equip it with a sound and complete equational theory that characterises when two models represent the same distribution.
\end{abstract}

\section{Introduction}
\label{sec:introduction}

Mixture models, particularly Gaussian mixture models (GMMs), are among the most widely used tools in probabilistic modelling. They serve as flexible approximations for complex multimodal distributions and appear in diverse applications, ranging from clustering and density estimation to generative modelling and signal processing~\cite{bishop2006pattern,mclachlan2000fmm,sung2004gaussian}. Moreover, it is well-known that any continuous distribution can be arbitrarily well approximated by a finite mixture of Gaussians~\cite{Goodfellow-et-al-2016}.   Despite their ubiquity and expressiveness, reasoning formally about mixtures, especially about when two mixture models with different structures define the same distribution, remains a subtle and largely unexplored problem. 

This paper develops a compositional, algebraic theory for models that include GMMs, allowing one to reason about their structure and equivalence. Our motivation stems from the observation that mixture models are not just statistical tools but structured expressions that admit algebraic transformations. For instance, a mixture with nested components can often be flattened or reparametrised without changing its semantics, and different symbolic expressions may describe the same underlying distribution. Yet these equivalences are rarely made explicit or given formal status.

The first step to address this gap is to define a formal syntax for a class of hybrid models that include discrete and Gaussian random variables. Our approach is grounded in the categorical framework of symmetric monoidal categories (SMCs)~\cite{maclane1971} and their associated language of \emph{string diagrams}~\cite{Selinger10}. In recent years, SMCs have been identified as a convenient algebraic setting in which to study how probabilistic models can be composed both in sequence and in parallel~\cite{Fritz_SyntheticApproach}. In this setting, string diagrams provide an intuitive and rigorous graphical syntax for reasoning about such models. Notably, structural features such as conditional independence, marginalisation, and factorisation can be expressed and manipulated directly in the diagrammatic language~\cite{Fritz_SyntheticApproach,fritz2023d,jacobs2018logicalessentialsbayesianreasoning}, making it particularly intuitive for human use. Unlike conventional graphical models, where diagrams are informal aids to understanding, our diagrams \emph{are} the syntax---they provide a finer-grained perspective on the internal structure of the distributions they represent and come with a formal, compositional semantics, given through a symmetric monoidal functor into a suitable category of probabilistic maps, assigning to each diagram a well-defined stochastic kernel that captures its intended probabilistic meaning. As a result, we can reason not only about the overall distribution but also about how probabilistic dependencies compose and combine in a modular and compositional manner.

Hybrid models that incorporate discrete and continuous random variables have been extensively studied in the context of probabilistic graphical models. A paradigmatic example is the family of \emph{Conditional Gaussian} distributions (or, simply, CG-distributions), a family of models in which the conditional distribution of continuous variables \emph{given} the discrete ones follows a Gaussian distribution~\cite{Lauritzen84,Lauritzen89}.  CG-distributions were first studied in the context of generalised graphical linear models,  associated to \emph{conditional linear Gaussian} networks \cite{Koller09}---Bayesian networks that combine discrete and continuous nodes, but in which discrete nodes cannot have continuous parents. Note also that GMMs are a specific case of conditional linear Gaussian networks, in which all of discrete variables are latent. 

The main result of this work is an algebraic axiomatisation of diagrammatic equivalence for this class of hybrid models: a \emph{sound and complete} set of equational rules that captures precisely when two diagrams represent the same distribution. That is, whenever two diagrams encode the same distribution, there is an equational derivation that transforms one into the other. In short, our contribution can be understood as providing a formal syntax, semantics, and complete equational theory for reasoning about CG-distributions in a modular syntax that refines that of conditional linear Gaussian networks. The benefit is two-fold. First, it offers a formal language for verifying that two models define the same distribution without having to compute the semantics explicitly, paving the way to implementation and automation. Second, it supports compositional reasoning: by breaking complex mixture models into parts, we can apply \emph{local} equational rules, much like in standard algebra. 

A key technical challenge in developing a diagrammatic framework for hybrid models is how to combine two fundamentally different kinds of structure: one that captures discrete (specifically, Boolean) probabilistic variables, and another tailored towards continuous Gaussian variables with conditional dependencies given by affine transformations. Addressing this challenges required us to 1) design a formal syntax in which the two types of variables could interact meaningfully; 2) develop an appropriate \emph{compositional} semantics capable of accomodating both discrete and continuous variables; and 3) identify a sound and complete set of algebraic rules that characterises all the non-trivial interactions between the two fragments. Our solution draws on and generalises two earlier diagrammatic axiomatisations, one for discrete~\cite{probcirc}, and another for Gaussian models~\cite{gqa}, integrating the two into a coherent setting for reasoning about hybrid systems.

\noindent \textbf{Outline.} In Section~\ref{sec:background}, we recall the notation, terminology, and mathematical background from (categorical) probability theory, and introduce the semantic domain of interest. Section~\ref{sec:syntax} gives the formal syntax of the diagrams used throughout, motivates its use for both discrete and continuous variables, and Section~\ref{sec:semantics} equips it with a formal semantics capturing their behaviour. In Section~\ref{sec:equational-theory}, we present an equational theory sound for the chosen semantics. Section~\ref{sec:axiomatisation} shows the completeness of this theory via a high-level normalisation argument. Finally, Section~\ref{sec:conclusions} concludes with a discussion of future work.

\section{Background on (Categorical) Probability Theory}
\label{sec:background}

We begin by reviewing the terminology and background from probability theory relevant to our development. As our main objects of study are \emph{mixtures of Gaussians}, which combine discrete and continuous components, we draw on some basic concepts of measure theory. We also cover some of the key ideas of categorical probability theory, building on basic notions of category theory, including that of a SMC.  Following this review of standard material, we present our first original contribution: the definition of a (symmetric monoidal) category in which discrete and Gaussian distributions coexist and can be composed in a coherent, unified framework. A passing familiarity with measurable spaces and functions is assumed. 

\subsection{Stochastic kernels}
\label{sec:stoch}

In what follows, we define the category $\Stoch$ of measurable spaces and stochastic kernels, following the presentation of Fritz~\cite[Section 4]{Fritz_SyntheticApproach}. Given two measurable spaces $(X,\Sigma_X)$ and $(Y,\Sigma_Y)$, a \emph{stochastic kernel} $f\colon (X,\Sigma_X)\distto(Y,\Sigma_Y)$ is a map
$$
f\colon \Sigma_Y\times X\to [0,1], \qquad (U,x)\mapsto f(U|x)
$$
such that $f(\cdot | x)\colon \Sigma_Y\to [0,1]$ is a probability distribution for every $x\in X$ and $f(U|\cdot)\colon X\to [0,1]$ is a measurable function for every $U\in\Sigma_Y$. The notation $f(U|x)$ indicates that we can think of $f$ as a probability distribution over $Y$, conditional on $X$. If $g\colon(Y,\Sigma_Y)\distto(Z,\Sigma_Z)$ is another stochastic kernel, $g\circ f$ is given by the Chapman-Kolmogorov equation:
$$
(g\circ f )(V|x) = \int_{y\in Y} g(V|y)f(dy|x)
 $$
 This defines another stochastic kernel, making measurable spaces and stochastic kernels into a category, $\Stoch$, with identity morphisms given by Dirac deltas:
 $$
 \id_X(U|x) = \Dirac{U}{x} = \begin{cases}
 								1 & \text{if $x\in U$},\\
 								0 & \text{otherwise}.
						 \end{cases}
 $$
We call the full subcategory of $\Stoch$ spanned by finite sets $\FinStoch$. Notice that when $X$ and $Y$ are finite sets, any stochastic kernel $f\colon (X, \mathcal{P}(X))\distto (Y, \mathcal{P}(Y))$ can be thought of as a (column-)stochastic matrix, \emph{i.e.}, a matrix whose columns sum to one; the composition $g\circ f$ of two stochastic kernels between finite sets then boils down to standard matrix multiplication.

Given two measurable spaces $(X_1,\Sigma_1)$ and $(X_2,\Sigma_2)$, we can form the product space $(X_1\times X_2,\Sigma_1\otimes \Sigma_2)$, where $\Sigma_1\otimes \Sigma_2$ is the $\sigma$-algebra generated by the rectangle subsets, \emph{i.e.}, those of the form $U_1\times U_2$ for $U_1\in\Sigma_1$ and $U_2\in\Sigma_2$. Moreover, given two measures $\mu_1$ and $\mu_2$ over $(X_1,\Sigma_1)$ and $(X_2,\Sigma_2)$ respectively, we can form the product measure $\mu_1\otimes \mu_2$ uniquely defined by $(\mu_1\otimes \mu_2)(U_1\times U_2) = \mu_1(U_1)\mu_2(U_2)$. This construction extends to stochastic kernels, whose product we characterise as follows on rectangle subsets:
$$
f_1\otimes f_2 \colon (X_1\times X_2, \Sigma_1\otimes \Sigma_2)\distto (Y_1\times Y_2, \Omega_1\otimes \Omega_2), \quad (U_1\times U_2 \mid x_1,x_2) \mapsto f_1(U_1|x_1)f_2(U_2|x_2)
$$
The unit for this monoidal structure is the singleton set $1= \{\bullet\}$ equipped with its powerset as $\sigma$-algebra. This equips $\Stoch$ with the structure of a \emph{symmetric monoidal category} $(\Stoch, \otimes, 1, \sigma)$, with symmetry $\sigma_{X,Y} \colon (X\times Y, \Sigma_X\otimes \Sigma_Y)\distto (Y\times X, \Sigma_Y\otimes \Sigma_X)$ defined by $\sigma(V\times U|x,y) = \delta_x(U)\otimes \delta_y(V)$. $\FinStoch$ is then a symmetric monoidal subcategory of $\Stoch$.

Beyond finite sets, in this work, we will only consider one other class of measurable spaces: the space $\R^n$ equipped with its Borel $\sigma$-algebra $\Borel{\R^n}$. Note that $\Borel{\R^m}\otimes\Borel{\R^n} \cong \Borel{\R^{m+n}}$. We will use this isomorphism implicitly whenever it is convenient. 

\subsection{Gaussian Probability}

We review here the basics of Gaussian distributions and their transformation under affine maps \cite{Fritz_SyntheticApproach,lauritzenjensen}. For a more detailed account, we refer the reader to the extensive literature on the subject~\cite{Billingsley,Rao_LinearStatisticalInference}. Gaussian distributions---aka \emph{normal distributions}---are one of the most fundamental classes of continuous probability distributions for real-valued random variables.

\textbf{Univariate Gaussians.} 
A random variable $\textbf{X}$ is \emph{Gaussian} with mean $\mu\in\mathbb{R}$ and variance $\sigma^2$ with $\sigma\in[0,\infty)$ if it has a probability density function given by
    $f(x\mid\mu,\sigma^2) = \big\{2\pi\sigma^2 \big\}^{-1/2}\exp\Big\{-\frac{(x-\mu)^2}{2\sigma^2} \Big\}$
  for $x\in\mathbb{R}$ with respect to the Lebesgue measure. This is typically denoted by $\textbf{X}\sim\mathcal{N}(\mu,\sigma^2)$. 
  
  \textbf{Multivariate Gaussians.} Normal distributions generalise naturally to high dimensional spaces by pushing forward products of independent \emph{standard normal distributions}, $\mathcal{N}(0,1)$ under some affine transformation.

  \begin{definition}[Multivariate Gaussian distribution]
    \label{def:mgaussian}
    A $k$-dimensional random vector $\textbf{X}$ follows a \emph{multivariate Gaussian distribution} (\emph{$k$-variate}) if every linear combination of its components is normally distributed. Equivalently, $\textbf{X}$ can be expressed as $\textbf{X}=A\cdot \textbf{Z} + \mu$, where $A\in\mathbb{R}^{k\times n}$ has rank $n$, $\mu\in\mathbb{R}^k$ and the random vector $\textbf{Z}$ has components $Z_1,\dots,Z_n\sim\mathcal{N}(0,1)$. We call $\mu$ the mean vector and $\Sigma := AA^T$ the \emph{covariance matrix} of  $\textbf{X}$, and write $\textbf{X}\sim\Normal{k}{\mu}{\Sigma}$.
  \end{definition}

   Similar to the univariate case, a multivariate Gaussian $\textbf{X}\sim\mathcal{N}_k(\mu,\Sigma)$ is fully characterised by its mean vector $\mu\in\mathbb{R}^k$ and its covariance matrix, $\Sigma\in\mathbb{R}^{k\times k}$, representing the expectation and the pairwise covariances of the variables. 
   Note that $\Sigma$ is necessarily \emph{positive semi-definite} and can be factored as $\Sigma=LL^T$ for some lower triangular matrix $L$ (a Cholesky decomposition of $\Sigma$). 
We also write $\Normal{k}{\mu}{\Sigma}(Y)$ for the probability of an event $Y\in\Borel{\R^k}$. Note that, when the covariance matrix is $0$, the corresponding measure is a Dirac at its mean, $\Dirac{\cdot}{\mu}$.

  \textbf{Gaussian maps.} We are interested in particularly simple stochastic kernels comprising an affine transformation with additive Gaussian noise. Formally, a Gaussian map is a stochastic kernel $f\colon \R^m\distto \R^n$ for which there exists an $n\times m$ matrix $A$, a vector $b\in\R^n$ and a positive semi-definite $n\times n$ matrix $\Sigma$ such that $f(\cdot | x)\sim \Normal{n}{Ax+b}{\Sigma}$\footnote{Here, we see $f(\cdot | x)$ as a random variable with value in $\R^n$.}; we write $f(x) = \Normal{n}{Ax+b}{\Sigma}$. The composite of two Gaussian maps in $\Stoch$ is still a Gaussian map: with $f$ as before, and $g\colon \R^n\distto \R^k$ given by $g(y) =  \Normal{k}{Cy+d}{\Theta}$, then $(g\circ f)(x) =  \Normal{k}{CAx +Cb+d}{C\Sigma C^T+\Theta}$. Thus, Gaussian maps form a subcategory of $\Stoch$~\cite{Fritz_SyntheticApproach}. In fact, they form a symmetric monoidal subcategory of $\Stoch$ with the monoidal product given as in Section~\ref{sec:stoch}: for two Gaussian maps $f_1(x) =  \Normal{n_1}{A_1x +b_1}{\Sigma_1}$ and $f_2(x) =  \Normal{n_2}{A_2x +b_2}{\Sigma_2}$, then for $x_1\in\R^{m_1}$
  and $x_2\in\R^{m_2}$, let
  $$
  (f_1\otimes f_2)\left(\begin{matrix}x_1\\x_2\end{matrix}\right) = \Normal{n_1+n_2}{ \left(\begin{matrix} A_1 & 0\\ 0 & A_2\end{matrix}\right)\left(\begin{matrix}x_1\\x_2\end{matrix}\right)+\left(\begin{matrix}b_1\\b_2\end{matrix}\right)}{\left(\begin{matrix}\Sigma_1 & 0\\0 & \Sigma_2\end{matrix}\right)}
  $$
  
  \subsection{Mixtures of Gaussians}
  \label{sec:gaussian-mixtures}
  
 In this work, we are concerned with mixtures: convex combinations of several component distributions. Intuitively, sampling from a mixture proceeds by first choosing a component based on the mixture weights, then sampling from the chosen component. 

  \begin{definition}[Mixture distribution]
    A \emph{mixture distribution} is a finite convex combination of distributions $\{\mu_i\}_{i=1,\dots,k}$ over the same measurable space $(X,\Omega)$, i.e., a measure $\mu$ defined by
    $
      \mu(E)=\sum_{i=1}^k p_i\cdot \mu_i(E),
    	\text{where } \sum_{i=1}^k p_i=1 \text{ and } p_i\geq 0, \text{ for a measurable set } E\in\Omega.
    $
  \end{definition}

  In what follows we will be concerned with measurable spaces $(\mathbb{R}^n,\mathcal{B}(\mathbb{R}^n))$ for some natural number $n$, equipped with their usual Borel $\sigma$-algebra $\mathcal{B}(\mathbb{R}^n)$ where the component distributions $\mu_i$ are all (multivariate) Gaussians---we call these \emph{mixtures of Gaussians}. Note that a mixture of Gaussians is different from taking the convex sum of several Gaussian random variables, which is still Gaussian. 
  In what follows, we will make extensive use of the following unique characterisation result. 
 \begin{restatable}{proposition}{uniqueparamsmultigaussians}
    \label{prop:mixture-uniqueness}
	Mixtures of Gaussians $\sum_i p_i\cdot \Normal{k}{\mu_i}{\Sigma_i}$ are uniquely determined by their parameters, \emph{i.e.}, the mixture weights $p_i$, component means $\mu_i$, and covariances $\Sigma_i$.
 \end{restatable}
  
  \subsection{Conditional Gaussian mixtures}
  \label{sec:mixgauss}
  
  The first original contribution of our paper is the definition of a SMC that unifies $\FinStoch$, the category of stochastic kernels between finite sets (specifically, arrays of Booleans), and $\Gauss$, the category of Gaussian maps.  Morphisms of this category will be stochastic kernels of type 
  $\Bool^p\otimes \R^m\distto\Bool^q\otimes \R^n$ satisfying certain conditions. Given such a kernel $f\colon\Bool^p\otimes \R^m\distto\Bool^q\otimes \R^n$, $a\in\Bool^p$, $x\in\R^m$, and $b\in \Bool^q$, let $f(\cdot | b,a,x)$ be the distribution over $\R^n$ obtained by conditioning $f$ on the discrete output component being equal to $b$. The morphisms of our category are those for which the conditional $f(\cdot | b,a,x)$  is a mixture of Gaussians (if the probability of obtaining $b$ is zero, the conditional may be defined arbitrarily as a mixture of Gaussians). This implies that we can find some finite set $I$ and stochastic kernels $\varphi\colon\B^p\distto I\otimes\B^q$, $f_i\colon \Bool^p\otimes \Real^m\distto\Real^n$, such that $f_i(a,x) = A_i(a)x + \Normal{n}{\mu_i(a)}{\Sigma_i(a)}$, for all $i\in I, a\in\Bool^p,x\in\Real^m$ and, for all $B\subseteq\Bool^q,Y\in\Borel{\Real^n}$, we have
  \begin{equation}
    \begin{aligned}
      \label{eq:mixgauss-map}
      f(B\times Y\mid a,x) &= \sum_{b\in B}\sum_{i\in I}\varphi(i,b|a)\cdot f_i(Y|a,x)\\
      &= \sum_{b\in B}\sum_{i\in I}\varphi(i,b|a)\cdot\mathcal{N}_n\left(A_i(a) x+\mu_i(a),\Sigma_i(a)\right)(Y)\quad
    \end{aligned}
  \end{equation}

  \begin{definition}[Conditional Gaussian mixture]
    \label{def:mixgauss-map}
 	We call \emph{conditional Gaussian mixture} (CG-mixture) a stochastic kernel $f\colon \Bool^p\otimes \R^m\distto\Bool^q\otimes \R^n$ expressible as in~\eqref{eq:mixgauss-map}. 
  \end{definition}
The core idea is that CG-mixtures represent joint distributions over discrete and continuous variables, such that the conditional distribution over the continuous variables---given fixed values of the discrete variables in its domain---is a Gaussian mixture. The following is a consequence of Proposition~\ref{prop:mixture-uniqueness}.
\begin{restatable}{proposition}{uniqueparamscgmixtures}
	\label{prop:CGM-unique-params}
	CG-mixtures are uniquely determined by their parameters, \emph{i.e.}, in the notation above, the set $I$ and the stochastic kernels $\varphi$, $f_i(a,x)$.
\end{restatable}
 In light of the previous proposition, we use the notation 
$
f(a,x) = \sum_{i\in I} \varphi(i,\cdot| a)\cdot \Normal{n}{A_i(a)x + \mu_i(a)}{\Sigma_i(a)}
$ for CG-mixtures.
Note that any Gaussian map can be seen as a CG-mixture (with a single component); similarly, any  stochastic kernel between finite sets can be seen as a CG-mixture, with trivial Gaussian components $\Normal{0}{\bullet}{()}$, the only distribution over $\R^0\cong 1 =\{\bullet\}$ with covariance the unique $0\times 0$ matrix. Finally, any mixture of Gaussians $\sum_i p_i \cdot \Normal{n}{\mu_i}{\Sigma_i}$ is a CG-mixture $f\colon 0\distto \R^n$ given by $f(\bullet,\bullet) = \sum_i\varphi(i,\cdot\mid \bullet)\cdot\Normal{n}{\mu_i}{\Sigma_i}$ where $\varphi(i,\cdot \mid\bullet) = p_i$.

A related definition has appeared in the work of Lauritzen, in which Bayesian networks with discrete and Gaussian nodes are considered~\cite{lauritzenjensen}. He calls \emph{Conditional Gaussian} a distribution over sets of discrete and continuous random variables, such that the conditional distribution of the continuous variables given the discrete ones is Gaussian. Our approach recasts these ideas categorically. Importantly, we work in a variable-free setting, where latent variables can be thought of as the index set of the mixture and the observable discrete variables as its inputs/outputs. This is also why our main semantic object of interest, CG-mixtures, are mixtures and not simply Gaussians: as we saw, composing stochastic maps involves marginalising over the intermediate variables (which become latent), so that composing conditional Gaussians in the sense of Lauritzen returns a mixture in general (and not necessarily a single Gaussian). 
\begin{restatable}{proposition}{cgmixturescompose}
	\label{prop:mixgauss-category}
	CG-mixtures are closed under composition and monoidal product in $\Stoch$.
\end{restatable}
\noindent As a result,  CG-mixtures form a symmetric monoidal subcategory of $\Stoch$ (symmetries and identities of the relevant type can readily be seen to be CG-mixtures). 
\begin{definition}[$\MixGauss$]
  \label{def:mixgauss}
  We call  $\MixGauss$ the symmetric monoidal subcategory of $(\Stoch, \otimes, \sigma)$ whose objects are measurable spaces of the form $\Bool^q\otimes \R^n$ for some $q,n\in\N$ and whose morphisms are CG-mixtures.\footnote{The reader may notice that $(\Bool^p\otimes \R^m)\otimes(\Bool^q\otimes \R^n)\neq \Bool^{p+q}\otimes \R^{m+n}$, but that the two sides are merely isomorphic. In what follows, we will behave as if they were equal, assuming implicitly that we can use the isomorphisms $\Bool^p\otimes\Bool^q \cong \Bool^{p+q}$ and $\R^m\otimes\R^n \cong \R^{m+n}$ wherever necessary. This defines a \emph{strict} SMC, which is what we need as semantics to our diagrammatic syntax, also a strict SMC.}
\end{definition}

\section{String diagrammatic syntax}
\label{sec:syntax}

Our main goal is to provide a sound and complete equational theory for CG-mixtures, which faithfully capture both discrete (Boolean) and continuous (Gaussian) variables, as well as their interaction. We approach this problem by first giving a modular syntax for these distributions in terms of string diagrams, a standard graphical language used to represent morphisms in SMCs. 
While we review some of the basics of string diagrams below, we refer the reader to \cite{Introdiagrams} for a more comprehensive introduction.

The two-dimensional syntax we use is a formal, graphical notation for morphisms of a \textit{coloured product and permutation category} (or simply, a \emph{prop}). A prop is a strict SMC whose objects are the set $C^*$ of words over a (finite) number of \emph{colours} $C$ and the monoidal product is given by concatenation.

The prop that will serve as our syntax $P_\Sigma$ is freely generated from a \emph{monoidal signature} $\Sigma$, a set of generating morphisms $g:v\to w$, through \emph{sequential} ($c\poi d$) and \emph{parallel} ($c\otimes d$) composition of generators of appropriate type, together with identities $\idone:c\to c$ for $c\in C$, symmetry $\sym:cd\to dc$ for $c,d\in C$ and empty generator $\idzero:\varepsilon\to \varepsilon$ (where $\varepsilon$ denotes the empty word). The two binary operations used to construct new terms, $(-\poi -):P_\Sigma\times P_\Sigma\to P_\Sigma$ and $(-\otimes -):P_\Sigma\times P_\Sigma\to P_\Sigma$ are the categorical composition and the monoidal product of the prop, respectively. The morphisms of $P_\Sigma$ are thus terms of a $C^*\times C^*$-sorted syntax quotiented by the laws of SMCs:
\begin{center}
  $c_1\otimes(c_2\otimes c_3) = (c_1\otimes c_2)\otimes c_3 \qquad (c_1\otimes c_2)\poi(d_1\otimes d_2)=(c_1\poi d_1)\otimes(c_2\poi d_2)$
  $(c\poi d)\poi e=c\poi(d\poi e) \qquad c\poi \idone = c = \idone \poi c \qquad \idzero\otimes c =c= c\otimes \idzero$
  $(\idone\otimes c)\poi\sym = \sym\poi (c\otimes \idone)\qquad\sym\poi \sym = \idone\otimes \idone$
\end{center}
where $c, d, e$ and $c_i$, $d_i$ range over $\Sigma$-terms of the appropriate type (omitted for clarity). These axioms state that the two forms of composition are associative and unital, that they satisfy a form of interchange law, and that the wire crossings behave as expected. Each $\Sigma$-term $c$ of type $v\to w$ will be graphically represented as a string diagram with labelled wires $\circuit{c}{v}{w}$. For $c:u\to v$, $d:v\to w$, $c_i:v_i\to w_i$, we depict $c\poi d$ as sequential composition and $c_1 \otimes c_2$ as parallel composition of $\Sigma$-terms of type $u\to w$ and $v_1v_2\to w_1w_2$,
\begin{center}
$
\InputIfFileExists{c-d-normal-font.tikz}{}{\input{./tikz/c-d-normal-font.tikz}}
 \;=\;
\InputIfFileExists{seq-compose-normal-font.tikz}{}{\input{./tikz/seq-compose-normal-font.tikz}}

\qquad 
\InputIfFileExists{c1xc2-normal-font.tikz}{}{\input{./tikz/c1xc2-normal-font.tikz}}
 \;=\;
\InputIfFileExists{par-compose-normal-font.tikz}{}{\input{./tikz/par-compose-normal-font.tikz}}
$
\end{center}

\noindent
The signature $\Sigma$ over which we generate our graphical syntax contains
\begin{itemize}
	\item two colours, $\Boolobj$ and $\Realobj$, whose identities  we depict respectively as $\Boolid$ and $\Realid$ respectively;
	\item the following generating morphisms:
	\begin{center}
		$\BoolCounit \mid \BoolComult \mid \Andgate \mid \Notgate \mid \Flip{p}$
		
		$\RealCounit \mid \RealComult \mid \RealWUnit \mid \RealMult \mid \Scalar{k} \mid \Foot \mid \NN$
		
		$\ifmix$
	\end{center}
\end{itemize}
\noindent
These can be seen as the constants of our language, and indeed will be the basic components of the diagrams that make up our syntax. Intuitively, our string diagrams will be freely formed much like conventional circuits, by wiring the above generators in sequence or in parallel, crossing wires (with the symmetries, $\textcolor{gray}{\sym}$, $\textcolor{black}{\sym}$, $\symp{style={draw={booltype}}}{style={draw={realtype}}}$, $\symp{style={draw={realtype}}}{style={draw={booltype}}}$, etc.) and making them as long as required (with identities, $\Boolid$ or $\Realid$). We refer to the resulting two-coloured prop as $\MixCirc$, and its morphisms as \emph{mixed circuits} or simply \emph{circuits}. We write $\Boolobj^p$ (resp. $\Realobj^n$) to denote a word containing $p$ (resp. $n$) successive $\Boolobj$ (resp. $\Realobj$). A  circuit
$
\CGterm{c}{p}{m}{q}{n}
$
will be depicted using a white box $\mixcircuit{c}{p}{q}{m}{n}$, in which the first $p$ input and $q$ output wires have Boolean type $\Boolobj$, while the remaining $m$ input and $n$ output wires denote the real type $\Realobj$. To simplify our syntax, we will sometimes avoid labelling wires explicitly: in this case, unlabelled thick wires represent an arbitrary number of wires (including potentially zero wires), while a normal-width wire indicates a single wire.

Although we have not yet defined a formal semantics for interpreting these, the colours and shapes hint at their intended meaning. The \textit{grey generators} resemble traditional Boolean circuit gates with the addition of a \emph{probabilistic gate}, $\Flip{p}$, which emits a $\mathsf{T}$ (true) with probability $p$ and a $\mathsf{F}$ (false) with probability $1-p$. 
The remaining grey generators correspond to the standard logical gates of conjunction ($\Andgate$) and negation ($\Notgate$), together with a copying gate, $\BoolComult$, that broadcasts its input to two output wires, and a terminating wire, $\BoolCounit$, that discards any input. 

We will refer to string diagrams in this fragment as \emph{(probabilistic) Boolean circuits}, and use grey boxes to represent a generic Boolean circuit $\Boolobj^q\to\Boolobj^p$. Note that such circuits have appeared in previous work, equipped with a sound and complete axiomatisation~\cite{probcirc}. 

Similarly, the \emph{black generators} are interpreted as processing real values. The preliminary intuition for the \emph{copier}, $\RealComult$, and \emph{discarder}, $\RealCounit$, is the same as their Boolean counterpart; the generator $\RealMult$ denotes addition, $\RealWUnit$ produces the value 0, $\Foot$ produces the value 1, and $\Scalar{k}$ multiply by the scalar $k\in\mathbb{R}$.  This set of generators has previously been used to give a compositional syntax for \emph{affine maps}~\cite{gla}. The additional generator $\NN$ is interpreted as sampling values randomly from a \emph{standard normal distribution}, $\mathcal{N}(0,1)$. Together, these allow us to represent Gaussian maps~\cite{gqa}, as we will explain in the following section. 
We will refer to string diagrams in this fragment as \emph{Gaussian circuits}, denote the corresponding prop as $\GaussCirc$, and use black boxes to represent generic circuits in this fragment. 

 
The last generator, $\ifmix$, serves as the interface between the Boolean and Gaussian fragments of our syntax. It behaves like an \emph{if-then-else} gate, which selects the value of its output wire based on the value of the discrete one, which we call the \emph{guard}: if it is true, the gate outputs the second input; if it is false, it outputs its third input. 

\begin{remark}
	We extend the syntactic sugar given by the above thick wires to represent multiple instances of these generators,
	
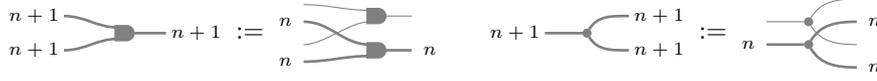
\begin{center}
          \begin{align*}
	    \label{eq:general-gates}
	    
\InputIfFileExists{and-1xn.tikz}{}{\input{./tikz/and-1xn.tikz}}
\; := \;
\InputIfFileExists{and-1xn-def.tikz}{}{\input{./tikz/and-1xn-def.tikz}}
\qquad 
\InputIfFileExists{copy-1xn.tikz}{}{\input{./tikz/copy-1xn.tikz}}
 \;:=\;
\InputIfFileExists{copy-1xn-def.tikz}{}{\input{./tikz/copy-1xn-def.tikz}}

          \end{align*}
	\end{center}
	\noindent
	and analogously with $\Notgate[thick]$ and $\BoolCounit[thick]$.

Finally, we will also make use of $n$-ary if-then-else circuits (for $n>1$), defined as the composite of $n$ if-then-else generators which share a single guard. Similarly, we will consider thickenned version of if-then-else circuits involving $p$ guards:
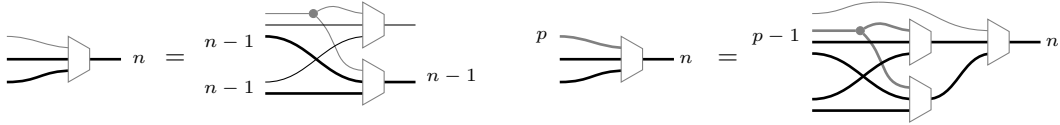
\begin{center}
  
\InputIfFileExists{thick-mixture-conditionals.tikz}{}{\input{./tikz/thick-mixture-conditionals.tikz}}

\end{center}
\end{remark}

\section{Semantics}
\label{sec:semantics}

We define the semantics of circuits as a mapping $\sem{\cdot}$ from the generators to stochastic kernels---specifically, to CG-mixtures. Since the diagrammatic syntax is freely-generated, the mapping will extend to a symmetric monoidal functor into $\MixGauss$ (Definition~\ref{def:mixgauss}), with which we can compute the behaviour of arbitrary composite circuits. 

Note that, for the purely Boolean or Gaussian fragments, our semantics coincides with that given in previous work: Boolean circuits are interpreted as stochastic kernels between sets of the form $\Bool^p$~\cite{probcirc}, while Gaussian circuits are interpreted as Gaussian maps~\cite{gqa}. 

We write $()$ for the unique $0\times 0$ matrix, $\bullet$ for the unique element of $\R^0$ and $\Bool^0$, and $\Dirac{\cdot}{x}$ for the Dirac delta at $x$. Each generator $\CGterm{g}{p}{m}{q}{n}$ is interpreted as a conditional Gaussian mixture $\sem{g}\colon \Bool^p\otimes\R^m\distto\Bool^q\otimes\R^n$, specified using the notation $\sem{g}(a,x) = \sum_{i\in I} \varphi(i,\cdot| a)\cdot \Normal{n}{A_i(a)x + \mu_i(a)}{\Sigma_i(a)}$ introduced in Section~\ref{sec:mixgauss}. Notice that none of our generators represents a non-trivial mixture of Gaussians---these will arise by composing them. Since for all the generators, either the Boolean or the real output component is trivial, we omit the irrelevant part of the expression, to avoid notational overload. For example, when $n=0$ (no real-valued output), we simply write $\sem{g}(a,x) = \sum_i\varphi(i,\cdot|a) = \sum_i p_i\cdot \Dirac{\cdot}{i}$ for some weights $p_i\in[0,1]$ that add to one. Moreover, when $q=0$ (no Boolean-valued output) and $I=1$ (trivial mixture), the resulting distribution is a single Gaussian map, which we write simply as $\sem{g}(a,x) = \Normal{}{Ax + \mu}{\Sigma}$. With these notational preliminaries out of the way, we are ready to specify the semantics of our generators:

\begin{figure}
  \begin{align*}
    &\sem{\;\BoolCounit\;}(b,\bullet) := \Dirac{\cdot}{\bullet} & & \sem{\;\BoolComult\;}(b,\bullet) := \Dirac{\cdot}{\begin{pmatrix}b\\b\end{pmatrix}} \\
    &\sem{\Andgate}\left(\begin{pmatrix}b_1\\b_2\end{pmatrix},\bullet\right) := \Dirac{\cdot}{b_1\land b_2} & & \sem{\Notgate}(b,\bullet) := \Dirac{\cdot}{\lnot b} \\
      &\sem{\;\Flip{p}\;}(\bullet,\bullet) := p\cdot \Dirac{\cdot}{1} + (1-p)\cdot \Dirac{\cdot}{0} & & \sem{\;\NN\;}(\bullet,\bullet) := \Normal{1}{0}{1}\\
      &\sem{\;\RealCounit\;}(\bullet, x) :=\Normal{0}{\bullet}{()} & & \sem{\;\RealComult\;}(\bullet, x) := \Normal{2}{\begin{pmatrix}x\\x\end{pmatrix}}{\begin{pmatrix}0 & 0\\0 & 0\end{pmatrix}}\\
      &\sem{\;\RealWUnit\;}(\bullet, \bullet) := \Normal{1}{0}{0} & & \sem{\;\RealMult\;}\left(\bullet, \begin{pmatrix}x_1\\x_2\end{pmatrix}\right):=\Normal{1}{x_1+x_2}{0}\\
        &\sem{\;\Scalar{k}\;}(\bullet, x) := \Normal{1}{k\cdot x}{0} & & \sem{\;\Foot\;}(\bullet,\bullet) := \Normal{1}{1}{0}
  \end{align*}
 \caption{Semantics of circuit generators in terms of CG-mixtures.}
\end{figure}
\begin{center}
	$
	\sem{\ifmix}\left(b,\begin{pmatrix}x_1\\x_2\end{pmatrix}\right):=  \begin{cases}
		\Normal{1}{x_1}{0} \text{ if } b = 1\\
		\Normal{1}{x_2}{0}  \text{ if } b = 0
	\end{cases}
	$
\end{center}
Recall also that the degenerate Gaussian $\Normal{n}{\mu}{0}$ is equal to $\Dirac{\cdot}{\mu}$, the Dirac at $\mu$. For example, $\sem{\;\Foot\;}(\bullet,\bullet):= \Normal{1}{1}{0}$ could have also been written as $\Dirac{\cdot}{1}$ (and the same goes for the interpretations of $\RealWUnit, \RealMult, \Scalar{k}, \RealComult, \RealCounit$).
The following is a simple consequence of the freeness of the diagrammatic syntax. 
\begin{proposition}
  $\sem{\cdot}$ extends to a symmetric monoidal functor $\MixCirc\to\MixGauss$.
\end{proposition}
This means that each circuit $\circuit{c}{v}{w}$ can be assigned a morphism $\sem{\circuit{c}{v}{w}}\colon  \Bool^p\otimes\R^m\distto\Bool^q\otimes\R^n$ of $\MixGauss$ where $p$ (resp. $q$) is the number of occurrences of $\Boolobj$ in $v$ (resp. $w$) and $m$ (resp. $n$) is the number of occurrences of $\Realobj$ in $v$ (resp. $w$) . Crucially, the previous proposition states that this assignment can be made \emph{compositional}---in other words, the semantics of any composite circuit can be computed from the semantics of the generators above following the expressions for composition and the monoidal product in $\MixGauss$ (which boil down to composition in $\Stoch$, \emph{cf.} Section~\ref{sec:background}). In particular, when $\sem{\circuit{c}{u}{v}}(a,x) = \sum_{i\in I} \varphi(i,\cdot| a)\cdot \Normal{m}{A_i(a)x + \mu_i(a)}{\Sigma_i(a)} $
and $\sem{\circuit{d}{v}{w}}(b,y) = \sum_{j\in J} \psi(j,\cdot| b)\cdot \Normal{n}{B_j(b)y + \nu_j(b)}{\Theta_j(b)}$ their composition is given by
$$\displaystyle
\sem{
\InputIfFileExists{seq-compose-normal-font.tikz}{}{\input{./tikz/seq-compose-normal-font.tikz}}
}(a,x) = \sum_{(i,j,b)\in I\times J\times B} \pi(i,j,b,\cdot| a)\cdot h_{i,j,b}(a,x)
$$
where $h_{i,j,b}(a,x) :=  \Normal{n}{B_j(b)A_i(a) x + B_j(b){\mu}_i(a) + \nu_j(b)}{B_ j(b)\Sigma_i(a)B_j(b)^T + \Theta_j(b)}$ and $\pi(i,j,b,\cdot |x):=\psi(j,\cdot|b)\varphi(i,b|a)$. Intuitively speaking, for every pair of Gaussian maps $f_i$ and $g_j$ of the two mixtures $\sem{c}$ and $\sem{d}$, we compose them \emph{qua} Gaussian maps (as explained in Section~\ref{sec:background}) and assign the product of the weights of their respective weights to the composite $g_j\circ f_i$ in the resulting mixture. 
Moreover, when
 $
 \sem{\circuit{c_1}{v_1\,}{\,w_1}}(a,x) = \sum_{i\in I} \varphi(i,\cdot| a)\cdot \Normal{m}{A_i(a)x + \mu_i(a)}{\Sigma_i(a)} $ and $ \sem{\circuit{c_2}{v_2\,}{\,w_2}}(b,y) = \sum_{j\in J} \psi(j,\cdot| b)\cdot \Normal{o}{B_j(b)y + \nu_j(b)}{\Theta_j(b)} 
 $
 their monoidal product $\sem{
\InputIfFileExists{par-compose-normal-font.tikz}{}{\input{./tikz/par-compose-normal-font.tikz}}
}\left(\!\left(\begin{matrix}\scriptstyle a\\\scriptstyle b\end{matrix}\right),\left(\begin{matrix}\scriptstyle x\\\scriptstyle y\end{matrix}\right)\!\right)$ is given by
 $$
 \!\!\!\!\sum_{(i,j)\in I\times J} \!\!\!\!\varphi(i,\cdot| a)\psi(j,\cdot| b) \cdot \Normal{m+o}{\!\!\left(\begin{matrix} \scriptstyle A_i(a) &\!\!\scriptstyle 0\\ \scriptstyle 0 & \!\!\scriptstyle B_j(b)\end{matrix}\right)\left(\begin{matrix}\scriptstyle x\\\scriptstyle y\end{matrix}\right)+\left(\begin{matrix}\scriptstyle\mu_i(a)\\\scriptstyle\nu_j(b)\end{matrix}\right)}{\left(\begin{matrix}\scriptstyle\Sigma_i(a) & \!\!\scriptstyle 0\\\scriptstyle 0 & \!\!\scriptstyle\Theta_j(b)\end{matrix}\right)\!\!}
 $$

\begin{example}\label{ex:mixture}
	We illustrate how to form mixtures of univariate Gaussians using this syntax. First, the following circuit allows us to take convex sum of its (real-valued) inputs:
	$$
	\sem{
\InputIfFileExists{convex-mix-if.tikz}{}{\input{./tikz/convex-mix-if.tikz}}
}\begin{pmatrix}x_1\\x_2\end{pmatrix} = p\cdot \Normal{1}{x_1}{0}  + (1-p)\cdot \Normal{1}{x_2}{0}  = p\cdot \Dirac{\cdot}{x_1} + (1-p)\cdot \Dirac{\cdot}{x_2}
	$$
	To form convex mixtures of Gaussians, we first need to understand how Gaussian circuits allow us to represent (multivariate) Gaussians, following~\cite{gqa}. As we saw, $\NN$ represents a single univariate centred and normalised Gaussian variable, $ \Normal{1}{0}{1}$. By multiplying it by some $\sigma\in\R$, we scale the variance by $\sigma^2$, obtaining a circuit $
\InputIfFileExists{NN-sigma.tikz}{}{\input{./tikz/NN-sigma.tikz}}
$ encoding $\Normal{1}{0}{\sigma^2}$; if we want to shift the mean by some value $\mu\in\R$, we can use $\RealMult$ to add $
\InputIfFileExists{mu-constant.tikz}{}{\input{./tikz/mu-constant.tikz}}
$. For example, 
	$$
	\sem{
\InputIfFileExists{gaussian-ex-1.tikz}{}{\input{./tikz/gaussian-ex-1.tikz}}
} = \Normal{1}{3}{1}\qquad \sem{
\InputIfFileExists{gaussian-ex-2.tikz}{}{\input{./tikz/gaussian-ex-2.tikz}}
} = \Normal{1}{0}{4}
	$$
	We can now take the $p$-mixture of the two circuits above:
	$$
	\sem{
\InputIfFileExists{mixture-gaussians-ex-1.tikz}{}{\input{./tikz/mixture-gaussians-ex-1.tikz}}
} = p\cdot \Normal{1}{3}{1}+(1-p)\cdot \Normal{1}{0}{4}
	$$
	For \emph{multivariate} Gaussians, we first need to recall how matrices (and vectors, which are just a special case) are encoded in graphical linear algebra~\cite{gla}.
  An $n \times m$ matrix $A$ may be represented by a circuit $
\InputIfFileExists{matrixcircuit.tikz}{}{\input{./tikz/matrixcircuit.tikz}}
$, using $\RealComult$, $\Scalar{x}$, $\RealMult$ as follows: its input wires stand for the columns of $A$, its output wires stand for the rows, and the $j$-th input wire is connected to the $i$-th output wire through a scalar $\Scalar{x}$ if the coefficient $A_{ij}$ is $x$. Moreover, since $\Scalar{0}$ is semantically equal to $\RealCounit\;\RealWUnit$ and $\Scalar{1}$ to $\Realid$, we can depict $A_{ij} = 0$ by not connecting the two input/output wires, and $A_{ij} =1$ by a plain wire. 
	\begin{center}
		$\text{For instance, for } A =  { \begin{pmatrix}
				\small x_1 & \small 0 & \small 0 \\
				\small 1 & \small 1 & \small 0 \\
				\small x_2 & \small 0 & \small 0 \\
				\small 0 & \small 0 & \small 0
		\end{pmatrix}}  \text{ let } 
\InputIfFileExists{matrixcircuit.tikz}{}{\input{./tikz/matrixcircuit.tikz}}
 := 
\InputIfFileExists{ex-matrix.tikz}{}{\input{./tikz/ex-matrix.tikz}}
 .$\end{center}
	Then, $\sem{
\InputIfFileExists{matrixcircuit.tikz}{}{\input{./tikz/matrixcircuit.tikz}}
}(x) = \Normal{n}{Ax}{0} = \Dirac{\cdot}{Ax}$, justifying the encoding. Now we can encode any multivariate Gaussian $\Normal{n}{\mu}{\Sigma}$ as the circuit
	$$
	
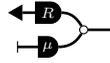
\InputIfFileExists{gaussian-nf.tikz}{}{\input{./tikz/gaussian-nf.tikz}}

	$$
	with $\Sigma = RR^T$, \emph{i.e.} some Cholesky decomposition of the covariance matrix.
        Then,
	$$
	 \sem{
\InputIfFileExists{2-mixture-gaussians.tikz}{}{\input{./tikz/2-mixture-gaussians.tikz}}
}= p\cdot \Normal{n}{\mu_1}{R_1R_1^T} + (1-p)\cdot \Normal{n}{\mu_2}{R_2R_2^T}
	$$
	Moreover, this can be iterated to mixtures with more components in a straightforward way. For example, for three components, we have
	$$
	\sem{
\InputIfFileExists{3-mixture-gaussians.tikz}{}{\input{./tikz/3-mixture-gaussians.tikz}}
} = p_1\cdot\Normal{n}{\mu_1}{R_1R_1^T} + p_2\cdot\Normal{n}{\mu_2}{R_2R_2^T} + p_3\cdot\Normal{n}{\mu_3}{R_3R_3^T}
	$$
	where $p = p_1$ and $q = \frac{p_2}{1-p_1}$.
	We will make use of this representation when defining a normal form for circuits below. 
\end{example}
From this example, it is easy to see that \emph{any} mixture of Gaussians can be encoded as a circuit; in fact, this is true of any CG-mixture: for $\sem{f}(a,x) = \sum_{i\in I} \varphi(i,\cdot| a)\cdot \Normal{n}{A_i(a)x + \mu_i(a)}{\Sigma_i(a)}$, we can encode $\varphi$ as a Boolean circuit and the individual Gaussian maps as Gaussian circuits~\cite{gqa}. Then, we can put these components together using if-then-else gates arranged in sequence as above, in order to select the right component for each Boolean input and each index $i\in I$.
\begin{restatable}{proposition}{universality}
	\label{prop:universality}
	For any CG-mixture $f$, there exists a circuit $c$ such that $\sem{c} = f$.
\end{restatable}

\section{Equational Theory}\label{equational-theory}
\label{sec:equational-theory}

While the laws of SMCs capture some of the semantic equivalences between circuits, they are not sufficient to derive all valid equivalences. To address this, we introduce a additional set of equations, ruling the interactions between the different components of our syntax. For more background on equational theories of string diagrams, aka symmetric monoidal theories, we refer the reader to a more detailed introductory text~\cite{Introdiagrams}. For a given theory $\mathsf{T}$, we write $\myeq{T}$ for an equality which can be derived by applying a finite sequence of axions of $\mathsf{T}$. Following standard terminology, we say that $\mathsf{T}$ is \emph{sound} (for a given semantics $\sem{\cdot}$) when, for any two circuits $c,d$ of the same type, $c\myeq{T} d$ implies that $\sem{c}=\sem{d}$. In other words, two circuits can be shown equal only if they denote the same stochastic kernel.  When the converse is true---$\sem{c}=\sem{d}$ implies $c\myeq{T} d$---we say that the theory is \emph{complete}. A sound and complete theory is said to \emph{axiomatise} the corresponding domain of interpretation. 

In this work, we are interested in the equational theory that rules the interaction of the Boolean and Gaussian fragments of our syntax, since each of these two individual fragments have already been axiomatised in previous work: probabilistic Boolean circuits~\cite{probcirc} and Gaussian circuits~\cite{gqa}. For this reason, we will always assume that two probabilistic Boolean (resp. Gaussian) circuits can be shown to be equal when they have the same interpretation, that is, when they denote the same stochastic kernel (appealing implicitly to the completeness results of the two cited papers). We will also assume that any stochastic kernel of type $\Bool^p\distto\Bool^q$ and any Gaussian map $\R^m\distto\R^n$ can be encoded as a circuit in the corresponding fragment---again, this follows from the cited results~\cite{probcirc,gqa}.

The main result of this paper is the soundness and completeness of the theory given in Fig.~\ref{fig:eqs} (relative to the completeness of the Boolean and Gaussian fragments). In what follows, we call this theory $\CGMTh$. We will justify its soundness in this section (Theorem~\ref{thm:soundness}); its completeness is more challenging and will be the object of the next section. 
%
%
\begin{figure}[!ht]
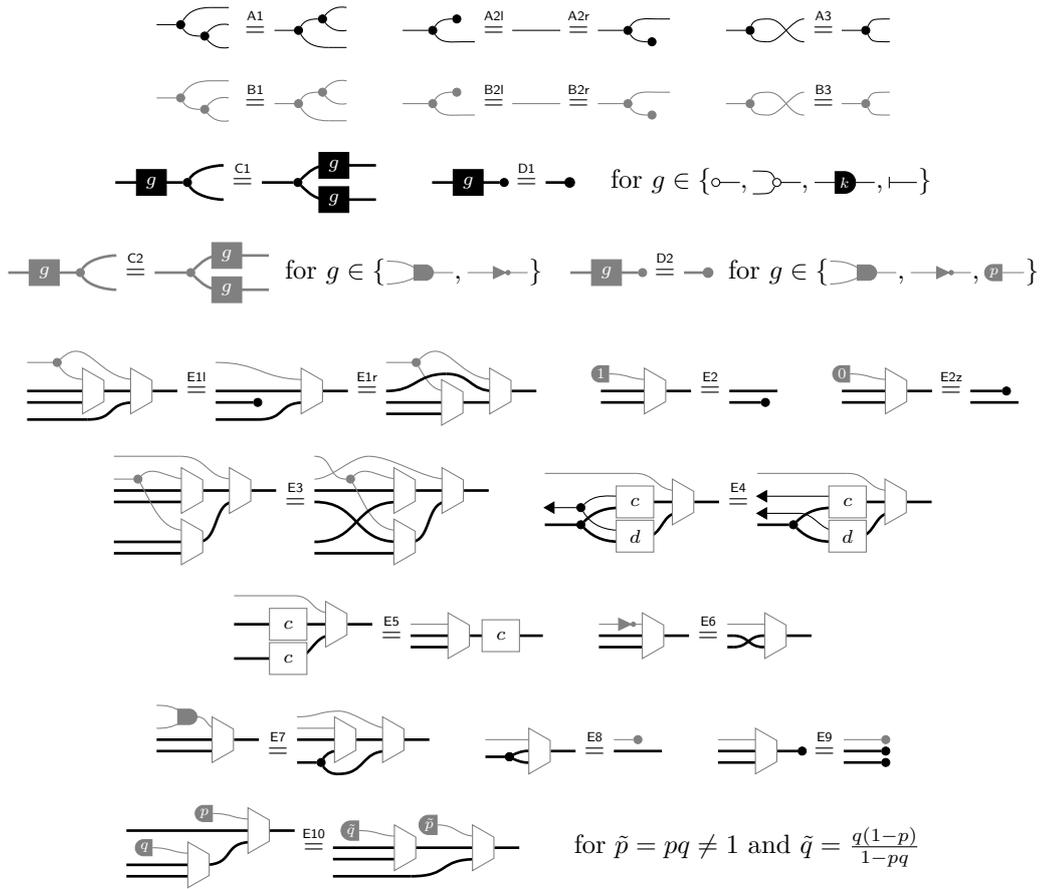

  {
    \centering
    \begin{center}
      $
\InputIfFileExists{axioms/real-copy-associative.tikz}{}{\input{./tikz/axioms/real-copy-associative.tikz}}
 \;\myeq{A1} 
\InputIfFileExists{axioms/real-copy-associative-1.tikz}{}{\input{./tikz/axioms/real-copy-associative-1.tikz}}
 \qquad 
\InputIfFileExists{axioms/real-copy-unital-left.tikz}{}{\input{./tikz/axioms/real-copy-unital-left.tikz}}
 \myeq{A2l} \idone\myeq{A2r} 
\InputIfFileExists{axioms/real-copy-unital-right.tikz}{}{\input{./tikz/axioms/real-copy-unital-right.tikz}}
\qquad  
\InputIfFileExists{axioms/real-copy-commutative.tikz}{}{\input{./tikz/axioms/real-copy-commutative.tikz}}
 \myeq{A3} \RealComult$
    \end{center}
    \begin{center}
      $
\InputIfFileExists{axioms/bool-copy-associative.tikz}{}{\input{./tikz/axioms/bool-copy-associative.tikz}}
 \;\myeq{B1} 
\InputIfFileExists{axioms/bool-copy-associative-1.tikz}{}{\input{./tikz/axioms/bool-copy-associative-1.tikz}}
 \qquad 
\InputIfFileExists{axioms/bool-copy-unital-left.tikz}{}{\input{./tikz/axioms/bool-copy-unital-left.tikz}}
 \myeq{B2l} \idonebool\myeq{B2r} 
\InputIfFileExists{axioms/bool-copy-unital-right.tikz}{}{\input{./tikz/axioms/bool-copy-unital-right.tikz}}
\qquad  
\InputIfFileExists{axioms/bool-copy-commutative.tikz}{}{\input{./tikz/axioms/bool-copy-commutative.tikz}}
 \myeq{B3} \BoolComult$
    \end{center}
    \begin{center}
      $
\InputIfFileExists{axioms/g-real-copy.tikz}{}{\input{./tikz/axioms/g-real-copy.tikz}}
\myeq{C1}
\InputIfFileExists{axioms/real-copy-gxg.tikz}{}{\input{./tikz/axioms/real-copy-gxg.tikz}}
 \qquad 
\InputIfFileExists{axioms/g-real-del.tikz}{}{\input{./tikz/axioms/g-real-del.tikz}}
\myeq{D1}\RealCounit[thick] \quad\text{ for } g\in\{\RealWUnit, \RealMult, \Scalar{k}, \Foot \}$
    \end{center}
    \begin{center}
      $
\InputIfFileExists{axioms/g-bool-copy.tikz}{}{\input{./tikz/axioms/g-bool-copy.tikz}}
\myeq{C2}
\InputIfFileExists{axioms/bool-copy-gxg.tikz}{}{\input{./tikz/axioms/bool-copy-gxg.tikz}}
\: \text{ for } g\in\{\Andgate, \Notgate\}\quad 
\InputIfFileExists{axioms/g-bool-del.tikz}{}{\input{./tikz/axioms/g-bool-del.tikz}}
\myeq{D2}\BoolCounit[thick]\:  \text{ for } g\in\{\Andgate, \Notgate, \Flip{p} \}$
    \end{center}
    \begin{center}
      $
\InputIfFileExists{axioms/axiom1lhs.tikz}{}{\input{./tikz/axioms/axiom1lhs.tikz}}
\myeq{E1l}
\InputIfFileExists{axioms/axiom1rhs.tikz}{}{\input{./tikz/axioms/axiom1rhs.tikz}}
\myeq{E1r}
\InputIfFileExists{axioms/axiom1bislhs.tikz}{}{\input{./tikz/axioms/axiom1bislhs.tikz}}
 \qquad 
\InputIfFileExists{axioms/axiom2lhs.tikz}{}{\input{./tikz/axioms/axiom2lhs.tikz}}
\myeq{E2}
\InputIfFileExists{axioms/axiom2rhs.tikz}{}{\input{./tikz/axioms/axiom2rhs.tikz}}
 \qquad 
\InputIfFileExists{axioms/axiom2bislhs.tikz}{}{\input{./tikz/axioms/axiom2bislhs.tikz}}
\myeq{E2z}
\InputIfFileExists{axioms/axiom2bisrhs.tikz}{}{\input{./tikz/axioms/axiom2bisrhs.tikz}}
$
    \end{center}
    \begin{center}
      $
\InputIfFileExists{axioms/axiom3lhs.tikz}{}{\input{./tikz/axioms/axiom3lhs.tikz}}
\myeq{E3}
\InputIfFileExists{axioms/axiom3rhs.tikz}{}{\input{./tikz/axioms/axiom3rhs.tikz}}
\qquad
\InputIfFileExists{axioms/axiom4lhs.tikz}{}{\input{./tikz/axioms/axiom4lhs.tikz}}
\myeq{E4}
\InputIfFileExists{axioms/axiom4rhs.tikz}{}{\input{./tikz/axioms/axiom4rhs.tikz}}
$
    \end{center}
    \begin{center}
      $
\InputIfFileExists{axioms/axiom5lhs.tikz}{}{\input{./tikz/axioms/axiom5lhs.tikz}}
\myeq{E5}
\InputIfFileExists{axioms/axiom5rhs.tikz}{}{\input{./tikz/axioms/axiom5rhs.tikz}}
\qquad
\InputIfFileExists{axioms/axiom6lhs.tikz}{}{\input{./tikz/axioms/axiom6lhs.tikz}}
\myeq{E6}
\InputIfFileExists{axioms/axiom6rhs.tikz}{}{\input{./tikz/axioms/axiom6rhs.tikz}}
$
    \end{center}
    \begin{center}
      $
\InputIfFileExists{axioms/axiom7lhs.tikz}{}{\input{./tikz/axioms/axiom7lhs.tikz}}
\myeq{E7}
\InputIfFileExists{axioms/axiom7rhs.tikz}{}{\input{./tikz/axioms/axiom7rhs.tikz}}
\qquad
\InputIfFileExists{axioms/axiom8lhs.tikz}{}{\input{./tikz/axioms/axiom8lhs.tikz}}
\myeq{E8}
\InputIfFileExists{axioms/axiom8rhs.tikz}{}{\input{./tikz/axioms/axiom8rhs.tikz}}
\qquad
\InputIfFileExists{axioms/axiom9lhs.tikz}{}{\input{./tikz/axioms/axiom9lhs.tikz}}
\myeq{E9}
\InputIfFileExists{axioms/axiom9rhs.tikz}{}{\input{./tikz/axioms/axiom9rhs.tikz}}
$
    \end{center}
    \begin{center}
      $
\InputIfFileExists{axioms/axiom10lhs.tikz}{}{\input{./tikz/axioms/axiom10lhs.tikz}}
\myeq{E10}
\InputIfFileExists{axioms/axiom10rhs.tikz}{}{\input{./tikz/axioms/axiom10rhs.tikz}}
 \qquad \text{for } \tilde p = pq \neq 1 \text{ and } \tilde q= \frac{q(1-p)}{1-pq}$
    \end{center}
  }
  \caption{Axioms of $\CGMTh$, the theory of Conditional Gaussian mixtures.}\label{fig:eqs} 
\end{figure}
%
A few comments are in order to clarify the meaning of the axioms of $\CGMTh$.
\noindent Axioms \textax{A1}-\textax{A3} encode the co-associativity, co-unitality, and co-commutativity of the copying ($\RealComult$) and discarding ($\RealCounit$) generators for real values.  These are common structure in many diagrammatic languages, guaranteeing that the different ways of sharing some value between different parts of a diagram are equal to one another. Immediately below, we have similar axioms \textax{B1}-\textax{B3} for the same operations on Boolean values.

The next two lines encode the interaction of the other generators with copying and discarding. Axioms \textax{C1}, \textax{C2} allow us to copy only those generators which denote \emph{deterministic} maps, \emph{i.e.}, those stochastic maps that, given their inputs, return a Dirac delta at their output. For example, $g = \RealMult$ represents addition, a deterministic operation that does not carry any randomness; unpacking \textax{C1} 
for this generator, we obtain
\begin{center}
  
\InputIfFileExists{axioms/wmult-copy.tikz}{}{\input{./tikz/axioms/wmult-copy.tikz}}
 $\myeq{C1}$ 
\InputIfFileExists{axioms/copyxcopy-wmultxwmult.tikz}{}{\input{./tikz/axioms/copyxcopy-wmultxwmult.tikz}}

\end{center}

\noindent Axioms \textax{D1}, \textax{D2} allow us to discard any generator: in other words, if we discard the output of some generator, we might as well discard its input. In the semantics, this corresponds to the fact that all the stochastic kernels that we can represent with our syntax correspond to \emph{normalised} (conditional) distributions.

The remaining axioms concern the interaction between the probabilistic behaviour of the discrete and the continuous parts of our syntax, mediated through the mixed \textit{if-then-else}. 

\noindent Axiom \textax{E1} captures the fact that the repeated evaluation of one condition inside of two nested if-then-else gates is redundant. Indeed, this axiom only rearranges if-then-elses, eliminating the test that depends on a shared Boolean guard, getting rid of the redundancy.

\noindent Axiom \textax{E2} is transparent in stating the simple fact that, whenever the guard of some conditional evaluates to true (resp. false), it results on always taking the \textit{then} (resp. \textit{else}) branch, allowing us to discard the other one. 

\noindent Axiom \textax{E3} states a form of distributivity of if-then-else over itself, provided that we swap the corresponding guards. This corresponds to a well-known property of if-then-else gates. 

\noindent Axiom \textax{E4} is an axiom scheme which holds for any circuits $c$ and $d$. It captures the intuitive idea that if the two branches of an if-then-else gate share a sample from the same Gaussian, we could also sample independently one Gaussian in each branch, since only one of them will be used.   
It can be seen as a form of commutativity of Gaussian sampling within a mixture-branching, allowing us to always sample locally when in presence of disjoint branches.

\noindent Axiom \textax{E5} is another axiom scheme, valid for any circuit $c$, which captures the distributivity of all operations over if-then-else gates: evaluating a function $f$ on two different inputs and later selecting only one of the outputs, is the same as first selecting the input and evaluating it on $f$. In categorical terms, if-then-else is \emph{natural}. Note that if $c=\RealComult$, this axiom also allows us to copy the if-then-else gate (like \textax{C1}, \textax{C2} allowed us to copy some other generators), and ensures that it is a deterministic operation. Similarly, if $c=\RealCounit$, this axioms allows us to discard if-then-else (like \textax{D1}, \textax{D2} allow us to discard all other generators); in other words, we can discard all three inputs of the if-then-else generator if we have discarded its output.

\noindent Axiom \textax{E6} encodes the simple fact that negating the guard of an if-then-else gate is equivalent to swapping the two branches. 

\noindent Axiom \textax{E7} states that a compound conditional with a conjunctive guard can be rewritten using nested if-then-else statements, where each Boolean variable is tested in sequence and the original branches are preserved accordingly.

\noindent Axiom \textax{E8} captures the fact that if the two branches of an if-then-else are the same, then we do not need an if-then-else in the first place (and can discard the guard).

\noindent Axiom \textax{E9} allows us to discard an if-then-else generator: similar to \textax{D1}, \textax{D2}, if we discard the output of some if-then-else, we might as well discard its guard and corresponding branches. 

\noindent Axiom \textax{E10} is the diagrammatic analogue of the well-known skew-associativity property of convex (aka barycentric) algebras~\cite{stone1949postulates}. To see this, recall (from Example~\ref{ex:mixture}) that the diagram 
 $ 
\InputIfFileExists{convex-mix-if.tikz}{}{\input{./tikz/convex-mix-if.tikz}}
$
behaves as a convex sum operation. Then, \textax{E10} allows us to re-associate different bracketings of convex sums, at the cost of changing the weights.

Crucially, our equational theory is \emph{sound}: any syntactic equality derivable from our axioms is a valid semantic equality. 
\begin{restatable}[Soundness]{theorem}{soundness}
	\label{thm:soundness}
	For all circuits, $c,d\colon v \to w$, if $c\CGMeq d$ then $\sem{c}=\sem{d}$.
\end{restatable}

\section{Completeness}
\label{sec:axiomatisation}
The main contribution of our work is the converse of Theorem~\ref{thm:soundness}, which states the \emph{completeness} of our theory for the chosen semantics. 

Our proof of completeness follows a normal form argument, where a syntactic representative is given for every semantic object in the image of the interpretation functor considered, later showing that when two circuits have the same semantics, they can always be transformed into their unique syntactic representative through a series of rewrite steps, following the equational axioms presented in Section~\ref{sec:equational-theory}. The core of the completeness argument is thus a normalisation proof: an explicit procedure to rewrite any given circuit into a normal form using only the axioms of our equational theory.

At a high-level, our completeness proof consists of two main steps, of which the first one is the most technically challenging. 

First, we explain how to normalise circuits of type $w\to \Realobj^n$, \emph{i.e.}, where the only possible outputs are real-valued random variables. Semantically, these correspond to mixtures of Gaussians (conditioned on their inputs, which can be both discrete or continuous). At a high level, the normalisation proof proceeds by structural induction: assuming that we have some circuit $d$ in normal form, we show how to normalise any circuit obtained by composing $d$ with any of the generators in our signature. 

Second, we address the normalisation of arbitrary circuits ---which can have both Boolean and real outputs--- by decomposing them into two parts: a purely Boolean subcircuit and another subcircuit of type $w\to \Realobj^n$. Finally, the completeness for Boolean circuits and the previous normalisation procedure for circuits of type $w\to \Realobj^n$ will allow us to obtain the general completeness result we are looking for. 

\begin{remark}\label{rmk:wlog}
	In what follows we will focus on circuits of type $\Boolobj^p\Realobj^m\to\Boolobj^q\Realobj^n$ without loss of generality. Indeed, for any arbitrary circuit $d\colon v\to w$, we can always pre-and post-compose it with the wire crossings ($\textcolor{gray}{\sym}$, $\textcolor{black}{\sym}$, $\symp{style={draw={booltype}}}{style={draw={realtype}}}$, $\symp{style={draw={realtype}}}{style={draw={booltype}}}$) to move all wires of type $\Boolobj$ before those of type $\Realobj$ (while preserving the order within the wires of each type), thereby obtaining a circuit $\CGterm{R(d)}{p}{m}{q}{n}$. This process is clearly reversible: it suffices to pre-and post-compose with the same symmetries in reverse to obtain the circuit $d$ with which we started. Moreover, for any two circuits $\CGterm{c,d}{p}{m}{q}{n}$, $c=d$ if and only if $R(c)=R(d)$. This is because, using the laws of SMCs, any axiom that we use to show that $c=d$ can be applied to show that $R(c)=R(d)$ and vice-versa.
\end{remark}

For clarity, and to simplify the overall proof, we break the definition of our normal form into parts. This highlights the interaction between the different components and allows us to define different normalisation procedures for circuits with different types, depending on whether they have Boolean inputs or outputs. First, the normal form for circuits with no Boolean inputs or outputs is a simple cascade of convex sums. 
\begin{definition}
  \label{def:cnf}
  A circuit $c\colon\Realobj^m\to\Realobj^n$ is in \emph{convex normal form} (CNF) if it is in the form
  \begin{center}
    
\InputIfFileExists{mixture-gaussians-nf.tikz}{}{\input{./tikz/mixture-gaussians-nf.tikz}}

  \end{center}
  where $p_0\in (0,1)$, $c_0:\Realobj^m\to \Realobj^n$ is in $\GaussCirc$, and $\GaussCircterm{c'}{m}{n}$ is in CNF or in $\GaussCirc$.
\end{definition}
\begin{definition}
  \label{def:mixgauss-nf}
  A circuit $d\colon\Boolobj^{p+1}\Realobj^m\to\Realobj^n$ is in \emph{normal form} (NF) if it is in the form inductively defined below,
  \begin{center}
    
\InputIfFileExists{cg-normal-form.tikz}{}{\input{./tikz/cg-normal-form.tikz}}

  \end{center}
  where \textsub{d}{i}, $i\in\{0,1\}$ are themselves in normal form, or in CNF (base case, see Definition~\ref{def:cnf}).
\end{definition}

Intuitively, \textsub{d}{1} corresponds with the circuit that results from $d$ conditioned on its first Boolean variable being true, and similarly with \textsub{d}{0}, where the same Boolean variable is assumed to be false instead.
The uniqueness of these normal forms stems from the unique characterisation of CG-mixtures by their parameters (Proposition~\ref{prop:CGM-unique-params}).
\begin{restatable}{proposition}{nfrealoutputunique}
	 \label{prop:nf-uniqueness}
	Any two circuits $c,d\colon\Boolobj^p\Realobj^m\to\Realobj^n$ in NF such that $\sem{c}=\sem{d}$, are equal.
\end{restatable}
  Intuitively, our normal form represents a circuit as a \emph{conditional probability distribution tree} \cite[Chapter~5]{Koller09}---a rooted tree where leaf nodes are labelled with distributions and interior nodes correspond to parent variables, with each outgoing edge associated with a unique variable assignment. This global structure, which captures the entire distribution, is mirrored by our circuits: convex combinations of (multivariate) Gaussians take the role of the leaves, and if-then-else gates ``split'' the tree by conditioning on discrete variables.

The axioms of $\CGMTh$ are sufficient to rewrite any given circuit without any Boolean outputs into one in normal form. The proof of this is the main technical contribution of our paper, and relies on a lengthy structural induction.
\begin{restatable}{theorem}{normalisationrealoutput}
  \label{thrm:completness-bprm-to-rn}
	Any circuit of type $\Boolobj^p\Realobj^m\to\Realobj^n$ is equal to one in normal form.
\end{restatable}


The normal form for arbitrary circuits below is a diagrammatic form of \emph{disintegration}, mirroring how a joint distribution $p(x,y)$ can be decomposed (that is, disintegrated) into the product $p(x)p(y|x)$ of a marginal and a conditional distribution. Any circuit $c$ can be decomposed in the same way as the composition of a Boolean circuit $b$ and a mixed circuit $d$ with only real-valued output. Semantically, the $\sem{d}$ is the CG-mixture obtained by conditioning $\sem{c}$ on  its Boolean output.

\begin{definition}
  \label{def:normal-form}
  A circuit $c\colon \CGtype{p}{m}{q}{n}$ is in \emph{normal form} (NF) if there exists some Boolean circuit $b$ and some mixed circuit $d$ in normal form (Definition~\ref{def:mixgauss-nf}) such that
  \begin{center}
    $\mixcircuit{c}{p}{q}{m}{n} = 
\InputIfFileExists{general-nf.tikz}{}{\input{./tikz/general-nf.tikz}}
$
  \end{center}
  which satisfy $\sem{d}(\cdot | a', a,x)= \Normal{n}{0}{0}$ if $\sem{b}(a'|a) = 0$, for $a\in\Bool^p, a'\in\Bool^q$, and $x\in\R^m$.
\end{definition}
The last condition is here to deal with edge cases in which we are conditioning on events of measure zero---a classic issue when defining disintegrations. When $\sem{b}(a'|a) = 0$, the conditional $\sem{d}$ is not well-defined when its Boolean input is $(a',a)$; therefore, any circuit will do to represent this case. However, to guarantee the uniqueness of the normal form, we need a convention: here, we choose $\Normal{n}{0}{0}$ or, equivalently, a Dirac concentrated at $0\in\R^n$.
\begin{restatable}{proposition}{nfunique}
	\label{prop:nf-circuits-equal}
	Normal forms are unique, \emph{i.e.}, for any two circuits $c,c'\colon\CGtype{p}{m}{q}{n}$ in normal form, if $\sem{c}=\sem{c'}$ then $c\CGMeq c'$. 
\end{restatable}
The following theorem relies extensively on Theorem~\ref{thrm:completness-bprm-to-rn} (normalisation of circuits without Boolean outputs); in this more general case, we simply need to check that the side condition of Definition~\ref{def:normal-form} (guaranteeing the uniqueness of the normal form) is satisfied. 
\begin{restatable}{theorem}{normalisationgeneral}
  \label{theorem:mixcirc-nf}
	Every circuit $\CGtype{p}{m}{q}{n}$ is equal to one in normal form.
\end{restatable}
\noindent Putting together the last two results, we can prove the completeness of our theory. 
\begin{theorem}[Completeness]
  For any two circuits $c,d$, if $\sem{c}=\sem{d}$ then $c\CGMeq d$.
\end{theorem}
\begin{proof}
	By Remark~\ref{rmk:wlog}, we can assume wlog that $c,d$ are two circuits of type $\CGtype{p}{m}{q}{n}$. By Theorem~\ref{theorem:mixcirc-nf}, these are equal to circuits $c'$ and $d'$ in normal form. Then, by soundness, if $\sem{c}=\sem{d}$, we also have $\sem{c'}=\sem{d'}$ and therefore $c'\CGMeq d'$, by Proposition~\ref{prop:nf-circuits-equal}, since both circuits are in normal form. Finally, by transitivity of equality, $c\CGMeq d$. 
\end{proof}

\section{Conclusions}
\label{sec:conclusions}

We presented a sound and complete axiomatisation of equivalence for Conditional Gaussian mixtures. We achieved this result through a calculus of string diagrams, whose semantics target morphisms in $\MixGauss$, a suitable category that combines discrete and Gaussian probability, as well as the complex interactions between them.

While probabilistic graphical models have long provided a flexible framework for specifying decision systems that operate in environments characterised by uncertainty \cite{Bernardo00,Koller09,lauritzenbook,Pearl88,Pearl09}, it is only in recent work that the notion of Markov categories \cite{Fritz_SyntheticApproach} and its diagrammatic syntax have been applied to the formalisation of these graphical models \cite{fong2013causaltheoriescategoricalperspective,jacobs2018logicalessentialsbayesianreasoning,jacobs2019causalinferencestringdiagram,lorenzin2025algebraicapproachmoralisationtriangulation}. Our work extends this line of research by giving a diagrammatic equational theory capable of capturing CG-mixtures, providing compositional semantics and a detailed image of the internal construction of the distribution represented.

We plan to extend this work in two orthogonal directions. First, we aim to extend our diagrammatic calculus with a primitive for \emph{conditioning}, enabling us to incorporate observed data into our models. An axiomatisation for CG-mixtures with conditioning would allow us to implement inference---particularly, parameter learning---as equational reasoning. 
Second, we plan to generalise our results to a broader classes of mixture models. Indeed, the only property of Gaussians on which our completeness proof depends is that mixtures of Gaussians are uniquely determined by their parameters (Proposition~\ref{prop:mixture-uniqueness}). This property underpins our definition of the normal form we use to show completeness. The rest of the proof is modular and can be adapted to other subcategories of stochastic maps, such as exponential or Poisson distributions, provided that similar characterisation results hold for their mixtures.

\newpage

\bibliography{refs}

\newpage

\appendix

\section{Appendix to Section~\ref{sec:background}}
\label{sec:app-cgp}
We restate the definitions and results from Section~\ref{sec:background}, providing proofs whenever these were omitted in the main text.

\begin{proposition}
  \label{prop:unique-params-univariate}
  Mixtures of univariate Gaussians are uniquely determined by their parameters.
\end{proposition}
\begin{proof}
  We write $f_{\mu,\sigma}$ as a shorthand for the probability density function of a univariate Gaussian distribution with mean $\mu$ and standard deviation $\sigma$. Let $p,q:X\to[0,1]$ be two finitely supported distributions on $X$. We want to show that for two arbitrary mixtures of Gaussians, $\sum_{i=1}^np_if_{\mu,\sigma,i}=\sum_{i=1}^mq_if_{\nu,\tau,i}$, these are equal if and only if $n=m$ and $(p,\mu,\sigma)_i=(q,\nu,\tau)_i$ for all $i\in\{1,\dots,n\}$.

  We will consider our mixtures to be increasing convex ordered. Intuitively, this entails that the first components in the mixture are the less spread and, amongst equally spread Gaussians, those having their means centered on lower values of $\Real$ are positioned first. Formally, given $f_{\mu_1,\sigma_1},f_{\mu_2,\sigma_2}$ two Gaussians, $f_{\mu_1,\sigma_1}$ is below $f_{\mu_2,\sigma_2}$ and we write it as $f_{\mu_1,\sigma_1}\leq f_{\mu_2,\sigma_2}$ if and only if $\sigma_1\leq\sigma_2$ and $\mu_1\leq\mu_2$ when $\sigma_1=\sigma_2$.

  We now examine what happens to the most spread-out component in the mixture, \emph{i.e.}, the last component according to our ordering, $p_nf_{\mu,\sigma,n}$, whenever its input, $x$, tends to infinity. To this end, we write the mixture as
  \[
  f(x) = \sum_{i=1}^n p_i\cdot\left(f_{\mu,\sigma}(x)\right)_i = p_n\cdot\left(f_{\mu,\sigma}(x)\right)_n+\sum_{i=1}^{n-1}p_i\cdot\left(f_{\mu,\sigma}(x)\right)_i
  \]
  and hence,
  \[
  \frac{f(x)}{\left(f_{\mu,\sigma}(x)\right)_n} = p_n + \sum_{i=1}^{n-1}p_i\left(\frac{\left(f_{\mu,\sigma}(x)\right)_i}{\left(f_{\mu,\sigma}(x)\right)_n}\right).
  \]

  We have that every $f_{\mu,\sigma,i}$ is a Gaussian and thus, the density of the above is given by
  \begin{align*}
    &= p_n + \sum_{i=1}^{n-1} p_i \left(\frac{\left\{2\pi\sigma_i^2\right\}^{1/2}\exp\left\{-\frac{1}{2}\frac{\left(x-\mu_i\right)^2}{\sigma_i^2}\right\}}{\left\{2\pi\sigma_n^2\right\}^{1/2}\exp\left\{-\frac{1}{2}\frac{\left(x-\mu_n\right)^2}{\sigma_n^2}\right\}}\right)\\
    &= p_n + \sum_{i=1}^{n-1} p_i \frac{\left\{2\pi\sigma_i^2\right\}^{1/2}}{\left\{2\pi\sigma_n^2\right\}^{1/2}}\exp\left\{-\frac{1}{2}\left(\frac{\left(x-\mu_i\right)^2}{\sigma_i^2}\right)-\left(-\frac{1}{2}\left(\frac{\left(x-\mu_n\right)^2}{2\sigma_n^2}\right)\right)\right\}\\
    &= p_n + \frac{1}{\left\{2\pi\sigma_n^2\right\}^{1/2}} \sum_{i=1}^{n-1} p_i\left\{2\pi\sigma_i^2\right\}^{1/2}\exp\left\{\frac{\left(x-\mu_i\right)^2}{2\sigma_i^2}-\frac{\left(x-\mu_n\right)^2}{\sigma_n^2}\right\}
  \end{align*}
  where, expanding the exponential in the last expression yields
  \[
  \exp\left\{\frac{(\sigma_i^2-\sigma_n^2)x^2+2(\mu_i\sigma_n^2-\mu_n\sigma_i^2)x+(\mu_n^2\sigma_i^2-\mu_i^2\sigma_n^2)}{2\sigma_i^2\sigma_n^2}\right\}.
  \]
  Hence, we have that the limiting value of the exponentials is governed by the difference $\sigma_i^2-\sigma_n^2$, which is less or equal to 0, given that the mixture is increasing convex ordered. Note that, if the standard deviation of these two components is the same, the exponential will be governed by the second factor, where $\mu_i\sigma_n^2-\mu_n\sigma_i^2=\mu_i-\mu_n<0$, and thus, the limit of the exponential is 0 for any component of the mixutre. From here,
  \[
  \lim_{x\to\infty}\frac{f(x)}{\left(f_{\mu,\sigma}(x)\right)_n}=p_n.
  \]
  Had we taken any pair of values other than $(\mu_n,\sigma_n)$, our limit would have diverged or gotten to 0. From here, $(\mu_n,\sigma_n)$ is the only parameter for which we converge to the weight of the n-th component, $p_n$. We then have that, given an arbitrary mixture of $n$ (distinct) Gaussian components, this is determined by its last component and, inductively, we get that if
  \[
  \sum_{i=1}^n p_if_{\mu,\sigma,i}=\sum_{i=1}^mq_if_{\nu,\tau,i}
  \]
  then $n=m$ and $(p,\mu,\sigma)_i=(q,\nu,\tau)_i$ for all $i\in\{1,\dots,n\}$.
\end{proof}

\uniqueparamsmultigaussians*
\begin{proof}
  This follows from the unique characterisation for mixtures of univariate Gaussians above. Let $\textbf{X}$ be a random vector in $\mathbb{R}^k$, a mixture of multivariate Gaussians given by $\sum_{i=1}^k p_i\cdot\mathcal{N}_k({\mu}_i,\Sigma_i)$. Then, for any vector in ${y}\in\mathbb{R}^k$, ${y}^T\cdot\textbf{X}\in\mathbb{R}^k$ follows a univariate mixture of Gaussians. We have that its parameters characterise the distribution of ${y}^T\textbf{X}$ uniquely by Proposition~\ref{prop:unique-params-univariate}. Since ${y}$ is arbitrary, we can recover each of the $\mu_i$ and $\Sigma_i$ uniquely.
\end{proof}

\uniqueparamscgmixtures*
\begin{proof}
  For a CG-mixture $f\colon \Bool^p\otimes \R^m\distto\Bool^q\otimes \R^n$, $f(\cdot | a,x)$ is a mixture of Gaussians for every $a\in \Bool^p$ and $x\in \R^m$. By Proposition~\ref{prop:mixture-uniqueness}, each such mixture is uniquely determined by its parameters $\varphi(i,\cdot|a)$, $A_i(a) x+\mu_i(a)$, and $\Sigma_i$. It is enough to evaluate the affine map $x\mapsto A_i (a)x+\mu_i(a)$ on some basis of $\R^m$ to recover the matrix $A_i(a)$ and vector $\mu_i(a)$ uniquely, thus proving the proposition. 
\end{proof}

\cgmixturescompose*
\begin{proof}
  Let $f\colon\Bool^p\otimes \R^\ell\distto\Bool^q\otimes \R^m$, $g\colon\Bool^q\otimes \R^m\distto\Bool^r\otimes \R^n$, be two CG-mixtures given by
  $$
  f(a,x) = \sum_{i\in I} \varphi(i,\cdot| a)\cdot f_i(a,x)\quad \text{and}\quad g(b,y) = \sum_{j\in J} \psi(j,\cdot| b)\cdot g_j(b,x)
  $$ 
  where $f_i(a) :=\lambda x. f_i(a,x)$ and $g_j(b):=\lambda y. g_j(b,y)$ are Gaussian maps.
  Then, for $a\in\Bool^\ell, x\in\R^\ell, C\subseteq\Bool^r,Z\in\Borel{\Real^n}$, we have
      \begin{align*}
  	&g\circ f(C,Z\mid a,x) = 
  	\\
  	&=\int_Y \sum_{b\in\Bool^q} g(C,Z\mid b,y) f(b,\diff y\mid a,x)
  	\\
  	&= \int_Y \sum_{b\in \Bool^q}  \sum_{c\in C}\left(\sum_{j\in J}\psi(j,c|b)\cdot g_j(Z|b,y)\right)  \left(\sum_{i\in I}\varphi(i,b|a)\cdot f_i(\diff y|a,x)\right)
  	\\
	&= \int_Y \sum_{b\in \Bool^q}  \sum_{c\in C}\sum_{j\in J}\sum_{i\in I}\psi(j,c|b)\cdot g_j(Z|b,y) \cdot \varphi(i,b|a)\cdot f_i(\diff y|a,x)
	\\
	&= \sum_{b\in \Bool^q}  \sum_{c\in C}\sum_{j\in J}\sum_{i\in I}\psi(j,c|b)\varphi(i,b|a)\cdot\left(  \int_Y g_j(Z|b,y)  f_i(\diff y|a,x) \right)
	\\
	&= \sum_{b\in \Bool^q}  \sum_{c\in C}\sum_{j\in J}\sum_{i\in I}\psi(j,c|b)\varphi(i,b|a)\cdot \big(g_j(b)\circ  f_i(a)\big)(Z| x) 
  \end{align*}
  where $\big(g_j(b)\circ  f_i(a)\big)$ denotes the composite of the Gaussian maps $g_j(b)$ and $f_i(a)$ in $\Stoch$. If 
  $$
  f_i(a,x) := \Normal{m}{A_i(a)x + \mu_i(a)}{\Sigma_i(a)} \text{ and } g_j(b) := \Normal{n}{B_j(b)y + \nu_j(b)}{\Theta_j(b)} 
  $$
  then let
  $$
  h_{i,j,b}(a,x) :=  \Normal{n}{B_j(b)A_i(a) x + (B_j(b){\mu}_i(a) + \nu_j(b))}{B_ j(b)\Sigma_i(a)B_j(b)^T + \Theta_j(b)}.
  $$
  Hence, the composition $g\circ f\colon\Bool^p\otimes \R^\ell\distto\Bool^r\otimes \R^n$ is given by\footnote{The tedious computation above should already convince the reader of the value of a formal syntax for CG-mixtures such as the one in this paper!}
  \[
  \displaystyle
  (g\circ f)(a,x) = \sum_{(i,j,b)\in I\times J\times B} \pi(i,j,b,\cdot| a)\cdot h_{i,j,b}(a,x), \text{ where }  \pi(i,j,b,\cdot |x):=\psi(j,\cdot|b)\varphi(i,b|a).
  \]
  
  \noindent Similarly, let $f\colon\Bool^p\otimes \R^\ell\distto\Bool^q\otimes \R^m$, $g\colon\Bool^r\otimes \R^n\distto\Bool^s\otimes \R^o$, given by
    $$
  f(a,x) = \sum_{i\in I} \varphi(i,\cdot| a)\cdot f_i(a,x)\quad \text{and}\quad g(b,y) = \sum_{j\in J} \psi(j,\cdot| b)\cdot g_j(b,x)
  $$ 
  where $f_i(a) :=\lambda x. f_i(a,x)$ and $g_j(b):=\lambda y. g_j(b,y)$ are Gaussian maps.
  Their monoidal product is given by
  $$
  (f\otimes g)\left(\left(\begin{matrix}a\\b\end{matrix}\right),\left(\begin{matrix}x\\y\end{matrix}\right)\right) = \sum_{(i,j)\in I\times J} \varphi(i,\cdot| a)\psi(j,\cdot| b) \cdot (f_i\otimes g_j)\left(\left(\begin{matrix}a\\b\end{matrix}\right),\left(\begin{matrix}x\\y\end{matrix}\right)\right) 
      $$
 If  the component Gaussian maps are given as below,
 	$$
     	f_i(a,x) := \Normal{m}{A_i(a)x + \mu_i(a)}{\Sigma_i(a)} \quad \text{and}\quad g_j(b,y) :=  \Normal{o}{B_j(b)x + \nu_j(b)}{\Theta_j(b)} 
     $$  
 then, $ (f\otimes g)\left(\left(\begin{matrix}a\\b\end{matrix}\right),\left(\begin{matrix}x\\y\end{matrix}\right)\right) = $
   $$
 \sum_{(i,j)\in I\times J} \varphi(i,\cdot| a)\psi(j,\cdot| b) \cdot \Normal{m+o}{\!\left(\begin{matrix} A_i(a) & 0\\ 0 & B_j(b)\end{matrix}\right)\left(\begin{matrix}x\\y\end{matrix}\right)+\left(\begin{matrix}\mu_i(a)\\\nu_j(b)\end{matrix}\right)}{\left(\begin{matrix}\Sigma_i(a)& 0\\0 & \Theta_j(b)\end{matrix}\right)\!}
 $$
\end{proof}

\section{Appendix to Section~\ref{sec:axiomatisation}}

\universality*
\begin{proof}
	Consider a CG-mixture $f\colon \Bool^p\otimes \Real^m\to \Bool^q\otimes\Real^n$ given by 
	$$f(a,x)=\sum_{i\in I} \varphi(i,\cdot| a)\cdot \Normal{n}{A_i(a)x + \mu_i(a)}{\Sigma_i(a)}$$

	The idea of this proof is to encode $f$ as the composition of two circuits: one Boolean circuit $b\colon\Boolobj^{p}\to\Boolobj^{|I|}\Boolobj^{q}$ which will represent $\varphi$, and one mixed circuit $d\colon\Boolobj^{|I|}\Realobj^m\to\Realobj^n$ that encodes the different Gaussian maps for each Boolean input and index value $i\in I$. In what follows, we use $|I|$ for the cardinality of the finite set $I$, and we will identify elements of $\Bool^{|I|}$ as subsets of $I$.
	
	First, let $b\colon\Boolobj^{p}\to\Boolobj^{|I|}\Boolobj^{q}$ be a Boolean circuit such that $\sem{b}(\{i\},a'|a) = \varphi(i,a'|a)$ and $\sem{b}(S,a'|a) = 0$ when $S\in\Boolobj^{|I|}$ does not encode a singleton. The intuition is that only one of the $|I|$ output wires of $b$ can be active at any given time (in other words, we use a one-hot encoding of the set $I$ as the singleton subsets of $\Bool^{|I|}$.
	
	Then, each component Gaussian map  $f_i(a,x):= \Normal{n}{A_i(a)x + \mu_i(a)}{\Sigma_i(a)}$ can be encoded as a Gaussian circuit using previous work~\cite{gqa}. 
	
	We have one of these circuits for each $a\in\Bool^p$ and each $i\in I$. To encode this dependence on $a,i$, we form a binary tree of if-then-else gates whose guards at each level are the $p+i$ incoming wires (as in Definition~\ref{def:mixgauss-nf}), and whose leaves are the Gaussian circuits corresponding to each component $f_i(a,x):= \Normal{n}{A_i(a)x + \mu_i(a)}{\Sigma_i(a)}$. For example, for $|I|=2$ and $p=1$, we have total of three layers of if-then-else gates organised as in the diagram below:
	\begin{center}
		
\InputIfFileExists{universality-tree.tikz}{}{\input{./tikz/universality-tree.tikz}}

	\end{center}
	The first two Boolean guard correspond to $\Boolobj^{|I|}$ and the next one to $\Boolobj^p$.
	 The choice of value (true or false) for each guard determines a path through the tree which corresponds to one possible input to the circuit and determines which Gaussian circuit (black box) will be selected. In the example above, the if we chose $i=1$ and set the input $a=\mathsf{F}$, this corresponds to the path leading to the circuit labelled $f_1(\mathsf{F})$, obtained by setting the first guard to true and the second to false encoding the mixture index $1$),  and the last one to false ($\mathsf{F}$). Recall that we are using a one-hot encoding of $I$ so that for guards that correspond to multiple $i\in I$ indices being true at the same time, we can just have any Gaussian circuit as input to the if-then-else gate; since these will never be selected when composed with $b$, their value is irrelevant (see below) and this is why we write them` as `$\dots$' above. 
	 
	 The result is a mixed circuit $d\colon\Boolobj^{|I|}\Realobj^m\to\Realobj^n$ such that 
	 $$
	 \sem{d}(\cdot | \{i\},a,x) = \Normal{n}{A_i(a)x+\mu_i(a)}{\Sigma_i(a)}
	 $$
	Finally, the circuit
	\begin{center}
		$c:= 
\InputIfFileExists{universal-b-d.tikz}{}{\input{./tikz/universal-b-d.tikz}}
$
	\end{center}
	satisfies $\sem{c}(a,x) = \sum_{i\in I} \varphi(i,\cdot| a)\cdot \Normal{n}{A_i(a)x + \mu_i(a)}{\Sigma_i(a)}$.
	
\end{proof}

\section{Appendix to Section~\ref{sec:semantics}}

\soundness*
\begin{proof}
	To prove the stated soundness, it is sufficient to show that, for any of the axioms above, the two sides denote the same conditional Gaussian mixture. 
	We show that the statement holds for $\textax{E4}$ and $\textax{E5}$---the other axioms can be checked in a similar fashion. 
	\begin{description}
		\item[- $\textax{E4}$:] Given that $c$ and $d$ have no Boolean inputs, we have
		\[\sem{c}(x,y)= \displaystyle\sum_{i\in I} p_i\cdot \Normal{m}{\!A_i\begin{pmatrix}x\\y\end{pmatrix} + \mu_i}{\Sigma_i}$ \text{ and }
		$\sem{d}(x,y) = \displaystyle\sum_{j\in J} q_j\cdot \Normal{n}{\!B_j\begin{pmatrix}x\\y\end{pmatrix} + \nu_j}{\Theta_j}\]
		Then
		\begin{align*}
			\sem{
\InputIfFileExists{axioms/axiom4lhs.tikz}{}{\input{./tikz/axioms/axiom4lhs.tikz}}
}(b,x) &= \begin{cases*}  
				\displaystyle\sum_{i\in I} p_i\cdot \Normal{m}{A_i\begin{pmatrix}0\\y\end{pmatrix} + \mu_i}{\Sigma_i}  \text{ if $b=1$}\\ 
				\displaystyle \sum_{j\in J} q_j\cdot \Normal{n}{B_j\begin{pmatrix}0\\y\end{pmatrix} + \nu_j}{\Theta_j}  \text{ if $b=0$}\end{cases*}
			\\
			&= \sem{
\InputIfFileExists{axioms/axiom4rhs.tikz}{}{\input{./tikz/axioms/axiom4rhs.tikz}}
}(b,x)
		\end{align*}
		\item[- $\textax{E5}$:] 
		\begin{align*}
			\sem{
\InputIfFileExists{axioms/axiom5lhs.tikz}{}{\input{./tikz/axioms/axiom5lhs.tikz}}
}\left(b,\begin{pmatrix}x_1\\x_2\end{pmatrix}\right) &= \begin{cases*}\sem{f}(x_1)\text{ if }b=1\\\sem{f}(x_2)\text{ if }b=0\end{cases*}\\
			&= \begin{cases*}\sem{f}\circ \Normal{n}{x_1}{0}\text{ if }b=1\\
				\sem{f}\circ \Normal{n}{x_2}{0}\text{ if }b=0\end{cases*} = \sem{
\InputIfFileExists{axioms/axiom5rhs.tikz}{}{\input{./tikz/axioms/axiom5rhs.tikz}}
}\left(b,\begin{pmatrix}x_1\\x_2\end{pmatrix}\right)
		\end{align*}
	\end{description} 
\end{proof}

\section{Appendix to Section 5}
\label{sec:app-axiomatisation}

\begin{restatable}{proposition}{cnfunique}
	\label{prop:cnf-uniqueness}
	Any two circuits $c,d\colon\Realobj^m\to\Realobj^n$ in CNF such that $\sem{c}=\sem{d}$, are equal.
\end{restatable}
\begin{proof}
	Let $c,d\colon\Realobj^m\to\Realobj^n$ be two circuits in CNF whose semantics is the same mixture of Gaussian maps $f = \sem{c}=\sem{d}$ where (since there is no Boolean input)
	\[f(x) =  \sum_i r_i \cdot f_i  = \sum_i r_i \cdot \Normal{n}{A_i x +\mu_i}{\Sigma_i}\] 
	Notice that, for any input $x\in\R^m$, $f(x)$ is a mixture of (multivariate) Gaussians, uniquely characterised by its parameters (weights $r_i$, means $A_i x$, and covariances $\Sigma_i$), by Proposition~\ref{prop:mixture-uniqueness}. Since this is true for any $x\in\R^m$, $f$ can be recovered uniquely from its parameters, \emph{i.e.}, the weights $r_i$ and the Gaussian maps $f_i$. 
	
	Moreover, since the two circuits are in CNF, we can unfold the normal form to obtain circuits $\{c_i\}_{0\leq i \leq M}$ and $\{d_j\}_{0\leq j \leq N}$ in $\GaussCirc$ such that
	\begin{align*}
		\sem{c} = p_0 \cdot \sem{c_0} + (1-p_0)\cdot \sem{c'} = p_0 \cdot \sem{c_0} + (1-p_0)\cdot\big(p_1\cdot \sem{c_1} + (1-p_1)\cdot \sem{c''}\big)= \dots\\
		\sem{d} = q_0 \cdot \sem{d_0} + (1-q_0)\cdot \sem{d'} = q_0 \cdot \sem{c_0} + (1-q_0)\cdot\big(q_1\cdot \sem{d_1} + (1-q_1)\cdot \sem{d''}\big) = \dots
	\end{align*}
	
	Since $\sem{c}=\sem{d} = f$ and $f=\sum_i r_i\cdot f_i$ is uniquely characterised by its weights and component maps, for each $c_i$, there must exist some Gaussian map $f_{k_i}$ such that $\sem{c_i}= f_{k_i}$; similarly, for each $d_i$, there exists some $f_{\ell_i}$ such that  $\sem{d_i}= f_{\ell_i}$. Therefore, $\{c_i\}_{0\leq i \leq M}$ and $\{d_j\}_{0\leq j \leq N}$ are equal \emph{as sets}. 
	
	We need to show that we can always reorder the components using the available axioms to make the two circuits equal. For this, it is sufficient to show that we can always swap any two neighbouring components (\emph{e.g.}, $c_i$ and $c_{i+1}$). This is what we show below:
	\begin{align*}
		
\InputIfFileExists{cnf-unfold-2.tikz}{}{\input{./tikz/cnf-unfold-2.tikz}}
 &\myeq{E10} 
\InputIfFileExists{cnf-unfold-2-swap-1.tikz}{}{\input{./tikz/cnf-unfold-2-swap-1.tikz}}

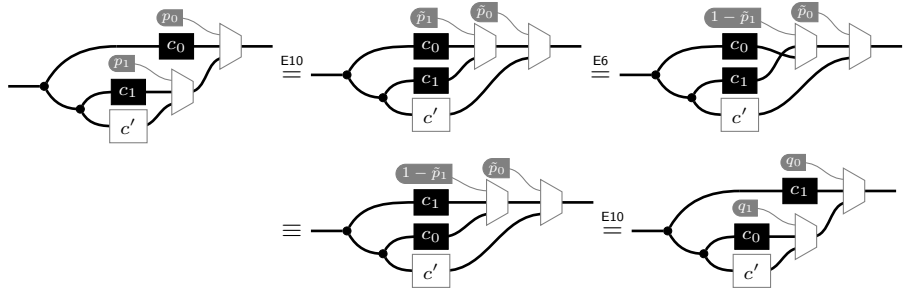
		\myeq{E6} 
\InputIfFileExists{cnf-unfold-2-swap-2.tikz}{}{\input{./tikz/cnf-unfold-2-swap-2.tikz}}
\\
		&\eqSMC 
\InputIfFileExists{cnf-unfold-2-swap-3.tikz}{}{\input{./tikz/cnf-unfold-2-swap-3.tikz}}

		\myeq{E10} 
\InputIfFileExists{cnf-unfold-2-swap-4.tikz}{}{\input{./tikz/cnf-unfold-2-swap-4.tikz}}

	\end{align*}
	where $\tilde{p}_0 = p_0p_1$, $\tilde q_1 = \frac{p_1(1-p_0)}{1-p_0p_1}$, and $q_0 = \tilde p_0(1-\tilde p_1)$, $q_1 = \frac{\tilde p_1(1-\tilde p_0)}{1- \tilde p_0(1-\tilde p_1)}$; note that these are in $(0,1)$ since $p_0,p_1$ are assumed to be in $(0,1)$.
	Thus, $\{c_i\}_{0\leq i \leq M}$ and $\{d_j\}_{0\leq j \leq N}$ are equal \emph{as multisets}.
	
	We now need to show that, whenever $c_i = c_j$ for $i\neq j$, then we can always collapse the two occurrences of $c_i$ in the $c$ into a single occurrence. Since, we can always reorder the components of the CNF, we can assume without loss of generality that $c_i$ and $c_j$ are neighbours, \emph{i.e.}, that $j=i+1$. Then,
	\begin{align*}
		
\InputIfFileExists{cnf-unfold-equal.tikz}{}{\input{./tikz/cnf-unfold-equal.tikz}}
 &\myeq{E10} 
\InputIfFileExists{cnf-unfold-equal-1.tikz}{}{\input{./tikz/cnf-unfold-equal-1.tikz}}

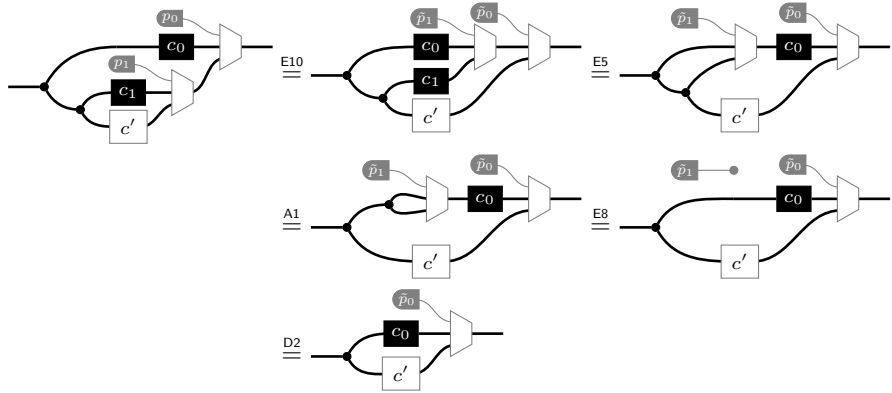
		\myeq{E5} 
\InputIfFileExists{cnf-unfold-equal-2.tikz}{}{\input{./tikz/cnf-unfold-equal-2.tikz}}
\\
		&\myeq{A1} 
\InputIfFileExists{cnf-unfold-equal-3.tikz}{}{\input{./tikz/cnf-unfold-equal-3.tikz}}

		\myeq{E8} 
\InputIfFileExists{cnf-unfold-equal-4.tikz}{}{\input{./tikz/cnf-unfold-equal-4.tikz}}
\\
		&\myeq{D2} 
\InputIfFileExists{cnf-unfold-equal-5.tikz}{}{\input{./tikz/cnf-unfold-equal-5.tikz}}

	\end{align*}
	Hence, $c_i = d_i$ for $0\leq i \leq M$ (and $M=N$). Finally, since the weights $r_i$ in $f=\sem{c}=\sem{d}$ are also unique, then we must have $p_i=q_i$ for  $0\leq i \leq M$, and therefore $c=d$. 
\end{proof}

\nfrealoutputunique*
\begin{proof}
	Let $c,d\colon\Boolobj^p\Realobj^m\to\Realobj^n$ be two circuits in NF whose semantics are the same as conditional Gaussian mixtures, \emph{i.e.}, stochastic kernels $f=\sem{c}=\sem{d}$ of type $\Bool^p\otimes\Real^m\distto\Real^n$ where
	\[
	f(a,x)=\sum_i\varphi(i|a)\cdot f_i(a,x)=\sum_i\varphi(i,\cdot |a)\cdot\Normal{n}{A_i(a)x+\mu_i}{\Sigma_i(a)}
	\]
	for $a\in\Bool^p$, and $x\in\Real^m$. 
	By Proposition~\ref{prop:CGM-unique-params}, $f(a,x)$ is uniquely determined by its parameters,  \emph{i.e.}, $\varphi$, as well as the Gaussian maps $\lambda x. f_i(a,x)$.
	
	Given that $c,d$ are circuits in NF, we can unfold them using our inductive definition,
		\begin{align*}
			\sem{c} &= \begin{cases}
				\sem{c|_1}(a',x) \text{ if } a_0=1\\
				\sem{c|_0}(a',x) \text{ if } a_0=0
			\end{cases}
		\end{align*}
	where $a_0$ is the first argument of $\varphi$, thus giving the only two possible restrictions of $\varphi|_{a_0}$; and similarly for $d$. Following Definition~\ref{def:normal-form}, we have that, given any possible bit-vector $a\in\Bool^p$, results in the same mixture of Gaussians, $f|a=\sem{c|a}=\sem{d|a}$, which, by Proposition~\ref{prop:cnf-uniqueness}, are uniquely determined by their parameters.
\end{proof}

\begin{lemma}
	\label{lemma:g-copy}
	For every generator $g\in\GaussCirc$ and circuits $d|_1,d|_0\colon\Boolobj^p\Realobj^m\to\Realobj^n$ $\MixCirc$ in normal form, the following equality is derivable in $\MixCirc$:
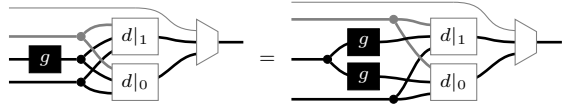
	\begin{center}
		
\InputIfFileExists{lemma2.tikz}{}{\input{./tikz/lemma2.tikz}}

	\end{center}
\end{lemma}
\begin{proof}
	The only case that cannot be dealt with using a single axiom is that of  $g=\NN$. We therefore focus solely on this case and proceed to show the statement of the lemma by induction on $p$, the number of Boolean input wires to \textsub{d}{1} and \textsub{d}{0}. The base case ($p=0$) follows directly from axiom \textax{E4}:
	\begin{center}
		
\InputIfFileExists{lemma2b-nn-01.tikz}{}{\input{./tikz/lemma2b-nn-01.tikz}}

	\end{center}
	
	\noindent We now assume that the lemma holds for a circuit with $p$ Boolean guards. We want to show that the equality is true for circuits with $p+1$ guards. We begin by unfolding \textsub{d}{1} and \textsub{d}{0}, using the fact that they are in normal form: 
	\begin{center}
		
\InputIfFileExists{lemma2b-nn-02.tikz}{}{\input{./tikz/lemma2b-nn-02.tikz}}

	\end{center}
	Thus, 
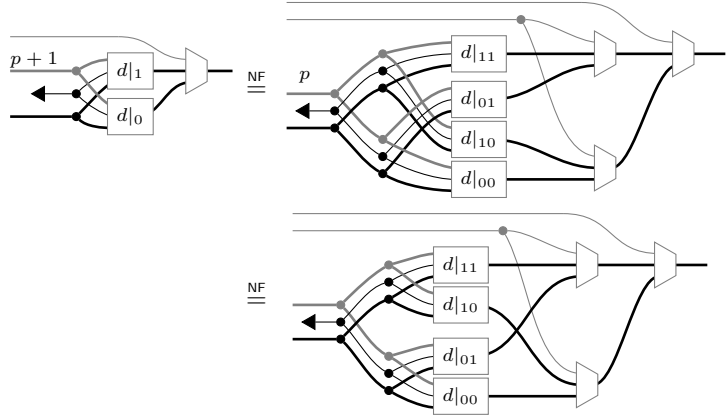
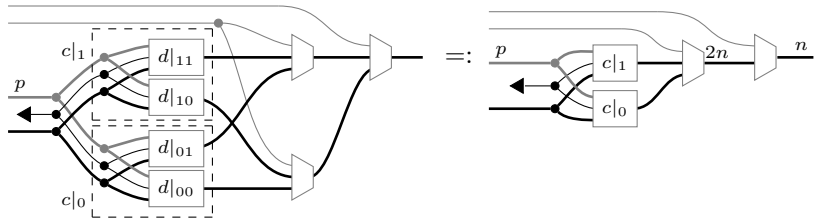
	\begin{center}
		
\InputIfFileExists{lemma2b-nn-03.tikz}{}{\input{./tikz/lemma2b-nn-03.tikz}}

	\end{center}
	where we use the definition of higher-arity if-then-else from Section~\ref{sec:syntax}, recalled below,
	\begin{center}
		
\InputIfFileExists{lemma2-nn-03.tikz}{}{\input{./tikz/lemma2-nn-03.tikz}}

	\end{center}
	and where we the rightmost if-then-else is depicted with a single thick wire representing its $2n=n+n$ real-valued input wires.
	
	Now, by the induction hypothesis (which we are entitled to apply since the \textsub{c}{1} and \textsub{c}{0} have $p$ Boolean input wires), we have
	\begin{center}
		
\InputIfFileExists{lemma2b-nn-04.tikz}{}{\input{./tikz/lemma2b-nn-04.tikz}}

	\end{center}
	Finally, repeating the previous steps in reverse, we obtain the desired result:
	\begin{center}
		
\InputIfFileExists{lemma2b-nn-05.tikz}{}{\input{./tikz/lemma2b-nn-05.tikz}}

	\end{center}
\end{proof}

The following lemma, 
can be seen as a proof that $\MixCirc$ quotiented by the equational theory presented in Section~\ref{sec:equational-theory} defines a \emph{Markov category}~\cite{Fritz_SyntheticApproach}. This amounts to showing  that discarding is natural in $\MixCirc$.
\begin{lemma}
	\label{lemma:discard-is-natural}
	Any circuit $c$ is discardable:
	\begin{center}
		
\InputIfFileExists{lemma-discard.tikz}{}{\input{./tikz/lemma-discard.tikz}}

	\end{center}
\end{lemma}
\begin{proof}
	We prove the lemma by structural induction on $c$. For the base case, it is enough to notice that each generator in our signature is discardable by the axioms, \textax{A2} for $\BoolComult$, \textax{B2} for $\RealComult$, \textax{D1} for $\Andgate$, $\Notgate$ and $\Flip{p}$, \textax{D2} for $\RealWUnit$, $\RealMult$, $\Scalar{k}$ and $\NN$, and \textax{E9} for $\ifmix$. For the inductive case, we have to show that the lemma holds for sequential and parallel composition, both of which follow immediatelly from the inductive hypothesis.
\end{proof}

\normalisationrealoutput*
\begin{proof}
	We proceed by structural induction. That is, we show that for an arbitrary  circuit $d:\Boolobj^p\Realobj^m\to\Realobj^n$ in normal form, precomposing it with any generator results in a circuit equal to one in normal form---all possible such composites fall into three categories:
	\begin{gather*}
		(i)\newline
		\noindent 
\InputIfFileExists{completeness-bprm-to-rn-shape1.tikz}{}{\input{./tikz/completeness-bprm-to-rn-shape1.tikz}}
 \qquad (ii) \:
\InputIfFileExists{completeness-bprm-to-rn-shape2.tikz}{}{\input{./tikz/completeness-bprm-to-rn-shape2.tikz}}
 \qquad (iii) \:
\InputIfFileExists{completeness-bprm-to-rn-shape3.tikz}{}{\input{./tikz/completeness-bprm-to-rn-shape3.tikz}}
 
	\end{gather*}
	
	\noindent \textit{(i)} \begin{minipage}{\widthof{$g:=\BoolComult\qquad\;$}}
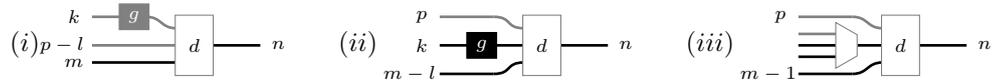
		\begin{mdframed}
			$g:=\BoolComult$
		\end{mdframed}
	\end{minipage} We start by unfolding twice the circuit $d$, which is in NF:
	\begin{center}
		
\InputIfFileExists{comp-bprm-to-rn-copy1.tikz}{}{\input{./tikz/comp-bprm-to-rn-copy1.tikz}}

	\end{center}
	
	\noindent At this point, we note that by the co-associativity of the copier, together with \textax{E1}, we have that
	\begin{center}
		
\InputIfFileExists{comp-bprm-to-rn-copy-prop1.tikz}{}{\input{./tikz/comp-bprm-to-rn-copy-prop1.tikz}}

	\end{center}
	
	\noindent And we can now apply Lemma~\ref{lemma:discard-is-natural} to conclude:
	\begin{center}
		
\InputIfFileExists{comp-bprm-to-rn-copy2.tikz}{}{\input{./tikz/comp-bprm-to-rn-copy2.tikz}}

	\end{center}
	
	\noindent \begin{minipage}{\widthof{$g:=\Andgate\qquad\;$}} 
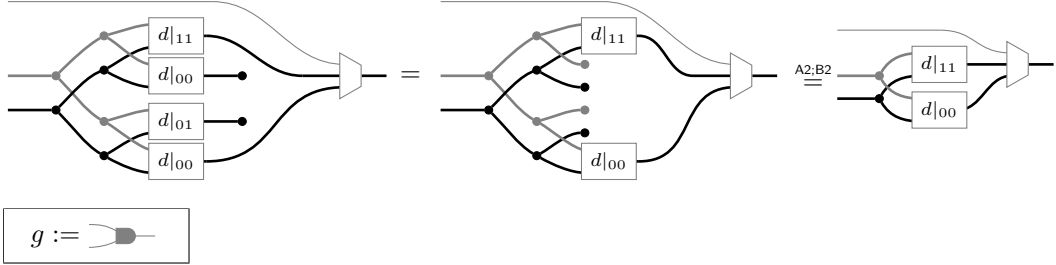
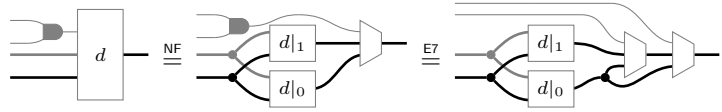
		\begin{mdframed}
			$g:=\Andgate$
		\end{mdframed}
	\end{minipage}
	
	\begin{center}
		
\InputIfFileExists{comp-bprm-to-rn-and1.tikz}{}{\input{./tikz/comp-bprm-to-rn-and1.tikz}}

	\end{center}
	\noindent We want to be able to push $d|_0$ past the copier to its right. For this, we need the following lemma.
	
	\begin{lemma}
		For circuits $d|_1,d|_0$ in NF, the following equality holds:
		
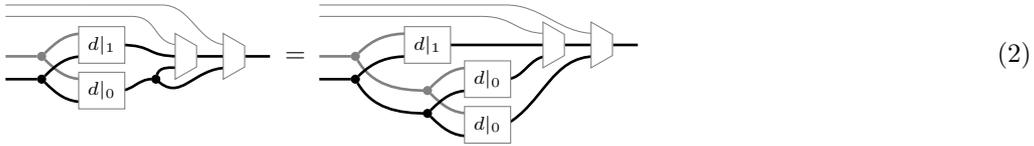
\begin{center}
			\begin{equation}
				\label{eq:comp-copy-nf}
				
\InputIfFileExists{comp-bprm-to-rn-and2.tikz}{}{\input{./tikz/comp-bprm-to-rn-and2.tikz}}

			\end{equation}
		\end{center}
	\end{lemma}
	\begin{proof}
		We prove this by induction on the number of Boolean inputs to $d\colon \Boolobj^p \Realobj^m\to\Realobj^n$, that is, on the number $p$. For the base case, $d$ is of type $\Boolobj^0\Realobj^m\to\Realobj^n$ and, thus, in CNF. We deal with this case in the lemma below. 
		
		\begin{lemma}
			For $c:\Realobj^m\to\Realobj^n$ in CNF, the following equality holds
			\begin{center}
				
\InputIfFileExists{comp-rm-to-rn-and-lemma01.tikz}{}{\input{./tikz/comp-rm-to-rn-and-lemma01.tikz}}

			\end{center}
		\end{lemma}
		\begin{proof}
			We reason by induction on the number of if-then-else in $c$. For the base case, we have that $c$ has no if-then-else and hence is a purely Gaussian circuit (in $\GaussCirc$). It is sufficient to show that, for every Gaussian circuit generator $g\in\Sigma_{\GaussCirc}$,
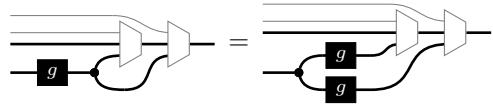
			\begin{center}
				
\InputIfFileExists{comp-rm-to-rn-and-lemma02.tikz}{}{\input{./tikz/comp-rm-to-rn-and-lemma02.tikz}}

			\end{center}
			\noindent This is immediate for all copiable generators, \emph{i.e.}, those that satisfy axiom \textax{C1}. The only $\GaussCirc$ generator which is not copiable is $\NN$; we check the equality nevertheless holds for this generator:
			\begin{center}
				
\InputIfFileExists{comp-rm-to-rn-and-lemma03.tikz}{}{\input{./tikz/comp-rm-to-rn-and-lemma03.tikz}}

			\end{center}
			\noindent where $f:=\Realid\otimes\Realid$, $g:=\RealCounit\otimes\RealComult$, as highlighted in the dashed boxes of the third circuit above. 
			
			For the inductive case, we assume that the equality of the lemma holds for circuits in CNF with $k$ if-then-else gates. Let $c$ be some circuit in CNF with $k+1$ if-then-else gates; we have
			\begin{align*}
				
\InputIfFileExists{and-proof-cnf-ih.tikz}{}{\input{./tikz/and-proof-cnf-ih.tikz}}
\;\myeq{CNF}\;	
\InputIfFileExists{and-proof-cnf-ih-1.tikz}{}{\input{./tikz/and-proof-cnf-ih-1.tikz}}
\;\myeq{E5}\;
\InputIfFileExists{and-proof-cnf-ih-2.tikz}{}{\input{./tikz/and-proof-cnf-ih-2.tikz}}

			\end{align*}
			
			After unfolding $c$ using the inductive definition for CNF circuits, we use \textax{B2} together with \textax{E8} to introduce an additional if-then-else, and use co-associativity (\textax{A1}) on the copied probabilistic guard, taking the circuit closer to a shape in which we can leverage our inductive hypothesis (IH) to copy $c$.
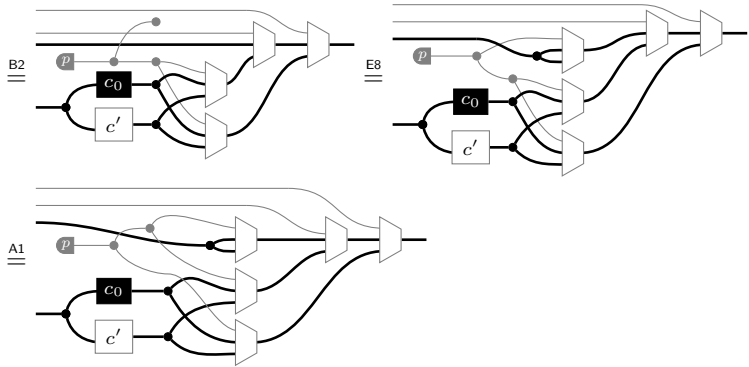
			\begin{center}
				\begin{align*}
					&\myeq{B2}
\InputIfFileExists{and-proof-cnf-ih-3.tikz}{}{\input{./tikz/and-proof-cnf-ih-3.tikz}}
\myeq{E8}
\InputIfFileExists{and-proof-cnf-ih-4.tikz}{}{\input{./tikz/and-proof-cnf-ih-4.tikz}}
\\
					&\myeq{A1}
\InputIfFileExists{and-proof-cnf-ih-5.tikz}{}{\input{./tikz/and-proof-cnf-ih-5.tikz}}

				\end{align*}
			\end{center}
			By re-associating the if-then-else gates, we can now use IH:
			\begin{gather*}
				\myeq{E3 x2}\quad
\InputIfFileExists{and-proof-cnf-ih-6.tikz}{}{\input{./tikz/and-proof-cnf-ih-6.tikz}}
\quad\myeq{IH x2}\quad
\InputIfFileExists{and-proof-cnf-ih-7.tikz}{}{\input{./tikz/and-proof-cnf-ih-7.tikz}}

			\end{gather*}
			Using naturality (\textax{E5}) and our associativity axiom (\textax{E3}), we get the circuit back intro a shape where we can remove the extra if-then-else introduced earlier.
			
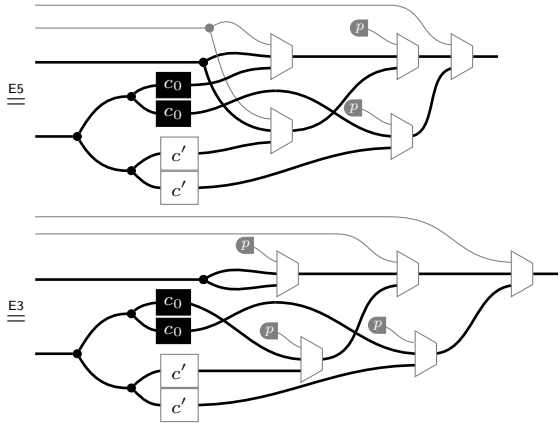
\begin{center}
				\begin{align*}
					&\myeq{E5}
\InputIfFileExists{and-proof-cnf-ih-8.tikz}{}{\input{./tikz/and-proof-cnf-ih-8.tikz}}
\\
					&\myeq{E3}
\InputIfFileExists{and-proof-cnf-ih-9.tikz}{}{\input{./tikz/and-proof-cnf-ih-9.tikz}}

				\end{align*}
			\end{center}
			Finally, we use \textax{E8} to get rid of the redundant mixed if-then-else, and the laws of SMCs to obtain two copies of the original circuit we started with.
			\begin{gather*}
				\myeq{E8}
\InputIfFileExists{and-proof-cnf-ih-10.tikz}{}{\input{./tikz/and-proof-cnf-ih-10.tikz}}
\myeq{SMC}
\InputIfFileExists{and-proof-cnf-ih-11.tikz}{}{\input{./tikz/and-proof-cnf-ih-11.tikz}}

			\end{gather*}
			It is clear that the right-hand side of the above equation corresponds with our target circuit:
			\begin{center}
				
\InputIfFileExists{and-proof-cnf-ih-12.tikz}{}{\input{./tikz/and-proof-cnf-ih-12.tikz}}

			\end{center}
		\end{proof}
		\noindent The above lemma proves the base case of Equation~\ref{eq:comp-copy-nf}, since it suffices to show that $d|_0$ can be copied whenever it is in CNF. The inductive step follows:
		\begin{center}
			
\InputIfFileExists{and-proof1.tikz}{}{\input{./tikz/and-proof1.tikz}}

		\end{center}
		\noindent We begin by unfolding $c$ using the inductive definition of our normal form. Next, using axiom \textax{E5} with $\RealComult[thick]$, we can copy the if-then-else gate. We then introduce an additional mixed if-then-else, which will later allow us to apply our inductive hypothesis:
		\begin{center}
			
\InputIfFileExists{and-proof2-1.tikz}{}{\input{./tikz/and-proof2-1.tikz}}

		\end{center}
		\noindent Reassociating the if-then-else gates using axiom \textax{E3}, we can apply the induction hypothesis:
		\begin{center}
			
\InputIfFileExists{and-proof2-2.tikz}{}{\input{./tikz/and-proof2-2.tikz}}

		\end{center}
		\noindent Redistributing the if-then-else generators back into their original form using the same \textax{E3} steps in reverse: 
		\begin{center}
			
\InputIfFileExists{and-proof4.tikz}{}{\input{./tikz/and-proof4.tikz}}

		\end{center}
		\noindent Finally, note that the right-hand side of the equation above is exactly the target circuit in our proof:
		\begin{center}
			
\InputIfFileExists{and-proof5.tikz}{}{\input{./tikz/and-proof5.tikz}}

		\end{center}
	\end{proof}
	
	\noindent Using the above, we now complete the normalisation proof for the $\Andgate$ generator. Starting from the right-hand side of Equation~\ref{eq:comp-copy-nf}:
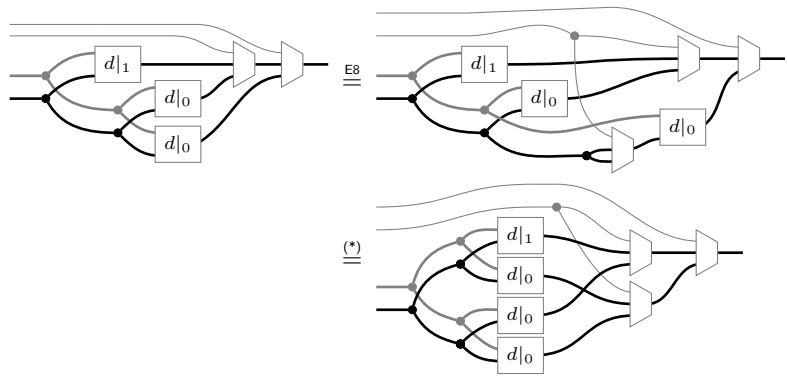
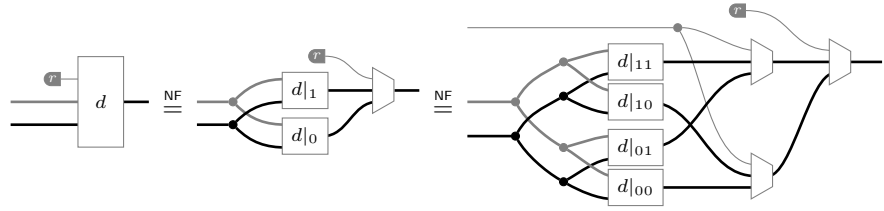
	\begin{center}
		
\InputIfFileExists{and-proof6.tikz}{}{\input{./tikz/and-proof6.tikz}}

	\end{center}
	
	\noindent To complete the proof, we justify the last step in the above, \emph{i.e.}, $\mathrel{(*)}$. For this, we show a modified version of axiom \textax{E5}. 
	\begin{lemma}
		For $c:\Boolobj^p\Realobj^m\to\Realobj^n$ an arbitrary circuit in NF, the following equality holds:
		\begin{center}
			
\InputIfFileExists{and-proof7.tikz}{}{\input{./tikz/and-proof7.tikz}}

		\end{center}
	\end{lemma}
	\begin{proof} By induction on the number of guards. The base case is given by axiom \textax{E5}. For the inductive step, we first apply \textax{E5}  with $c:= \RealComult$:
		\begin{center}
			
\InputIfFileExists{and-proof8.tikz}{}{\input{./tikz/and-proof8.tikz}}

		\end{center}
	We can then apply the IH, re-organise the $c|_i$s using the laws of SMCs, and use co-associativity of $\BoolComult$  to conclude:
		\begin{center}
			
\InputIfFileExists{and-proof9.tikz}{}{\input{./tikz/and-proof9.tikz}}

		\end{center}
	\end{proof}
	
	\noindent \begin{minipage}{\widthof{$g:=\Notgate\qquad\;$}} 
		\begin{mdframed}
			$g:=\Notgate$
		\end{mdframed}
	\end{minipage}
	Since $d$ is in NF, we can expand it once as in the first step below, and apply axiom \textax{E6} in the second step to recover a circuit in NF:
	
	\begin{center}
		
\InputIfFileExists{comp-bprm-to-rn-not1.tikz}{}{\input{./tikz/comp-bprm-to-rn-not1.tikz}}

	\end{center}
	
	\noindent \begin{minipage}{\widthof{$g:=\Flip{p}\qquad\;$}} 
		\begin{mdframed}
			$g:=\Flip{r}$
		\end{mdframed}
	\end{minipage}
	We reason by induction on the number $p$ of Boolean input wires of
	\begin{equation}
		\label{eq:flip-p-d-nf}
		
\InputIfFileExists{flip-p-d-nf.tikz}{}{\input{./tikz/flip-p-d-nf.tikz}}

	\end{equation}
	For $p=0$, the resulting circuit is in CNF, so we are done. Assume that~\eqref{eq:flip-p-d-nf} is equal to a circuit in NF for all $d$ with $p>1$ Boolean input wires; consider $d$ with $p+1$ Boolean inputs. 
	Once again, since $d$ is in NF, we can expand it twice as below:
	\begin{center}
		
\InputIfFileExists{comp-bprm-to-rn-flip1.tikz}{}{\input{./tikz/comp-bprm-to-rn-flip1.tikz}}

	\end{center}
	
	\noindent Applying \textax{E5} (with $c:= 
\InputIfFileExists{axiom7-natural-f.tikz}{}{\input{./tikz/axiom7-natural-f.tikz}}
$) to the last circuit we obtain
	\begin{center}
		
\InputIfFileExists{comp-bprm-to-rn-flip2.tikz}{}{\input{./tikz/comp-bprm-to-rn-flip2.tikz}}

	\end{center}
	With a little re-organisation (using the laws of SMCs), we see that the two sub-circuit in the dashed boxes below in the middle are in NF; since they have $p$ Boolean input wires, we can invoke the induction hypothesis to find $d'_1$ and $d'_0$ in NF such that
	\begin{center}
		
\InputIfFileExists{comp-bprm-to-rn-flip3.tikz}{}{\input{./tikz/comp-bprm-to-rn-flip3.tikz}}

	\end{center}
	Since $d'_1$ and $d'_0$ are in NF, so is the last circuit, as we wanted.
	
	\noindent \textit{(ii)}   We want to show that for an arbitrary circuit $d\colon\Boolobj^p\Realobj^m\to\Realobj^n$ in NF, composing with any of the generators of $\GaussCirc$ results in a circuit that is equal to one in NF, \emph{i.e.} that
	\begin{equation}\label{eq:idxgxid-d}
		
\InputIfFileExists{completeness-bprm-to-rn-shape2.tikz}{}{\input{./tikz/completeness-bprm-to-rn-shape2.tikz}}

	\end{equation}
	for any generator $g\colon k\to l$ can be shown equal to a circuit in NF. 
	
	We show this by induction $p$. For the base case ($p=0$), $d$ is a circuit in CNF and therefore we need to show that
	\begin{equation}\label{eq:gxid-d}
		
\InputIfFileExists{completeness-gxid-d.tikz}{}{\input{./tikz/completeness-gxid-d.tikz}}

	\end{equation}
	is equal to a circuit in CNF. To show this, we need another induction---this time on the number of if-then-else generators in $d$ (since it is in CNF). For the base case, $d$ contains no if-then-else and~\eqref{eq:gxid-d}
	is a Gaussian circuit and thus already in CNF. Now, assume that we can show that circuit~\eqref{eq:gxid-d} is equal to one in CNF for any $d$ in CNF with $n$ if-then-else generators. Let $d$ be some in CNF with $n+1$ if-then-else. Then, by Lemma~\ref{lemma:g-copy}, we have that for all generators $g\in\GaussCirc$, we have
	\begin{center}
		
\InputIfFileExists{completeness-rm-to-rn-01.tikz}{}{\input{./tikz/completeness-rm-to-rn-01.tikz}}

	\end{center}
	Furthermore, given  that each $d|_i$ is itself in CNF with $n$ if-then-else generators, we can apply the induction hypothesis and obtain circuits $d'|_i$ in CNF such that 
	\begin{center}
		
\InputIfFileExists{completeness-rm-to-rn-02.tikz}{}{\input{./tikz/completeness-rm-to-rn-02.tikz}}

	\end{center}
	where the last circuit is in CNF, as we wanted. 
	
	Now (for the inductive step of the induction on the number of Boolean input wires) assume that we can show that~\eqref{eq:idxgxid-d} is equal to a circuit in NF for any $d$ in NF with $p$ Boolean input wires. Let $d$ be some circuit in NF with $p+1$ Boolean input wires. Then, by Lemma~\ref{lemma:g-copy}, we have that for all generators $g\in\GaussCirc$, the following holds:
	\begin{center}
		
\InputIfFileExists{completeness-bprm-to-rn-gauss-01.tikz}{}{\input{./tikz/completeness-bprm-to-rn-gauss-01.tikz}}

	\end{center}
	Given that each $d|_i$ is itself in NF with $p$ Boolean input wires, we can apply the induction hypothesis and obtain circuits $d'|_i$ in NF such that 
	\begin{center}
		
\InputIfFileExists{completeness-bprm-to-rn-gauss-02.tikz}{}{\input{./tikz/completeness-bprm-to-rn-gauss-02.tikz}}

	\end{center}
	where the last circuit is now in NF, as we wanted. 
	
	\noindent \textit{(iii)} Once more, we reason by induction on the number of Boolean input wires of the circuit $d$. We start with the inductive case and will deal with the base case--which will require another induction--below. Assume that for every $d$ in NF with $p$ Boolean input wires, the circuit below is equal to one in NF:
	\begin{equation}
		\label{eq:mixif-d-nf}
		
\InputIfFileExists{mixif-d-nf.tikz}{}{\input{./tikz/mixif-d-nf.tikz}}

	\end{equation}
	Then, let $d$ be some circuit in NF with $p+1$ Boolean input wires. We can apply axiom \textax{E5} (with $f:=
\InputIfFileExists{thick-real-copy.tikz}{}{\input{./tikz/thick-real-copy.tikz}}
$) to copy the leftmost if-then-else as in the second step below:
	\begin{center}
		
\InputIfFileExists{comp-bprm-to-rn-mixif1.tikz}{}{\input{./tikz/comp-bprm-to-rn-mixif1.tikz}}

	\end{center}
	Now, since $d|_1$ and $d|_0$ are in NF with $p$ Boolean input wires, we can apply the induction hypothesis to obtain $d'|_1$ and $d'|_0$ in NF such that
	\begin{center}
		
\InputIfFileExists{comp-bprm-to-rn-mixif2.tikz}{}{\input{./tikz/comp-bprm-to-rn-mixif2.tikz}}

	\end{center}
	where the last circuit is therefore in NF, as we wanted.
	
	We now turn back to the base case: without any Boolean inputs, $d$ is necessarily in CNF and therefore we need to reason by induction again to normalise it---this time on the number of if-then-else generators it contains. Assume that for any $d$ in CNF with $p$ if-then-else generators, the circuit~\eqref{eq:mixif-d-nf} is equal to one in NF. Let $d$ be a circuit in CNF with $p+1$ if-then-else generators. Then, we can apply the same strategy as before. First, using axiom \textax{E5} (with $f:=
\InputIfFileExists{thick-real-copy.tikz}{}{\input{./tikz/thick-real-copy.tikz}}
$), we can copy the if-then-else as in the second step below:
	\begin{center}
		
\InputIfFileExists{comp-bprm-to-rn-mixif3.tikz}{}{\input{./tikz/comp-bprm-to-rn-mixif3.tikz}}

	\end{center}
	Then, since each $d|_i$ is in CNF with $p$ if-then-else generators, we can invoke the induction hypothesis to obtain $d'|_1$ and $d'|_0$ in NF such that
	\begin{center}
		
\InputIfFileExists{comp-bprm-to-rn-mixif3-nf.tikz}{}{\input{./tikz/comp-bprm-to-rn-mixif3-nf.tikz}}

	\end{center}
	For the base case (a CNF with \emph{no} if-then-else generators) $d$ is simply a circuit in $\GaussCirc$. In this case, we use (in order) co-unitality \textax{B2}, axiom \textax{E8}, and axiom \textax{E5} (naturality of if-then-else) to slide $d$ to the two branches of if-then-else and obtain a circuit in NF:
        \begin{gather*}
          
\InputIfFileExists{comp-bprm-rn-mixif01.tikz}{}{\input{./tikz/comp-bprm-rn-mixif01.tikz}}
\myeq{B2}
\InputIfFileExists{comp-bprm-rn-mixif02.tikz}{}{\input{./tikz/comp-bprm-rn-mixif02.tikz}}
\myeq{E8}
\InputIfFileExists{comp-bprm-rn-mixif03.tikz}{}{\input{./tikz/comp-bprm-rn-mixif03.tikz}}
\myeq{E5}
\InputIfFileExists{comp-bprm-rn-mixif04.tikz}{}{\input{./tikz/comp-bprm-rn-mixif04.tikz}}

        \end{gather*}
\end{proof}

\begin{proposition}
  \label{prop:copy-flip-ite}
  The following equality holds:
  \noindent
  
  \begin{center}
    
\InputIfFileExists{copied-flip-prop.tikz}{}{\input{./tikz/copied-flip-prop.tikz}}

  \end{center}
\end{proposition}
\begin{proof}
  \begin{gather*}
    
\InputIfFileExists{copied-flip-prop-01.tikz}{}{\input{./tikz/copied-flip-prop-01.tikz}}
\myeq{E8}
\InputIfFileExists{copied-flip-prop-02.tikz}{}{\input{./tikz/copied-flip-prop-02.tikz}}
\myeq{E5}
\InputIfFileExists{copied-flip-prop-03.tikz}{}{\input{./tikz/copied-flip-prop-03.tikz}}
\\
    \myeq{E8}
\InputIfFileExists{copied-flip-prop-04.tikz}{}{\input{./tikz/copied-flip-prop-04.tikz}}
\myeq{E3 x2}\quad
\InputIfFileExists{copied-flip-prop-05.tikz}{}{\input{./tikz/copied-flip-prop-05.tikz}}
\\
    \myeq{E5 x3}\quad
\InputIfFileExists{copied-flip-prop-06.tikz}{}{\input{./tikz/copied-flip-prop-06.tikz}}
\myeq{E5}
\InputIfFileExists{copied-flip-prop-07.tikz}{}{\input{./tikz/copied-flip-prop-07.tikz}}

    \\
    \myeq{E8}
\InputIfFileExists{copied-flip-prop-08.tikz}{}{\input{./tikz/copied-flip-prop-08.tikz}}

  \end{gather*}
\end{proof}
The following is just completeness for circuits with no Boolean outputs. 
\begin{theorem}\label{thm:completeness-no-bool-output}
	For any two circuits $c,d\colon\Boolobj^p\Realobj^m\to\Realobj^n$,$\sem{c}=\sem{d}$ implies $c\CGMeq d$. 
\end{theorem}
\begin{proof}
	First, using Theorem~\ref{thrm:completness-bprm-to-rn}, we can find $c'$ and $d'$ in NF such that $c\CGMeq c'$ and $d\CGMeq d'$. Moreover, by soundness of $\CGMTh$, $\sem{c'} = \sem{d'}$ and therefore $c'\CGMeq d'$ by uniqueness of normal forms (Proposition~\ref{prop:nf-uniqueness}). By transitivity, we conclude $c\CGMeq d$.
\end{proof}

\nfunique*
\begin{proof}
	First, if $\sem{c}=\sem{c'}$ then, by assumption, we have
	\begin{equation}\label{eq:nf-equal}
		$ 
\InputIfFileExists{general-nf.tikz}{}{\input{./tikz/general-nf.tikz}}
 =  
\InputIfFileExists{general-nf-prime.tikz}{}{\input{./tikz/general-nf-prime.tikz}}
$
	\end{equation}
	We want to show that $b=b'$ and $d=d'$. For the former, by composing with $\RealCounit$, we obtain
	\begin{center}
		$\sem{
\InputIfFileExists{general-nf.tikz}{}{\input{./tikz/general-nf.tikz}}
} ; \sem{
\InputIfFileExists{bidxrdel.tikz}{}{\input{./tikz/bidxrdel.tikz}}
} = \sem{
\InputIfFileExists{general-nf-prime.tikz}{}{\input{./tikz/general-nf-prime.tikz}}
} ; \sem{
\InputIfFileExists{bidxrdel.tikz}{}{\input{./tikz/bidxrdel.tikz}}
}$	
	\end{center}
	By functoriality of $\sem{\cdot}$, we have
	\begin{center}
		$\sem{
\InputIfFileExists{nf-rdel.tikz}{}{\input{./tikz/nf-rdel.tikz}}
} = \sem{
\InputIfFileExists{nf-prime-rdel.tikz}{}{\input{./tikz/nf-prime-rdel.tikz}}
}$	
	\end{center}
	and (by Lemma~\ref{lemma:discard-is-natural} or by a similar semantic argument) we can discard $d$, giving us
	\begin{center}
		$\sem{
\InputIfFileExists{nf-rdel-1.tikz}{}{\input{./tikz/nf-rdel-1.tikz}}
} = \sem{
\InputIfFileExists{nf-prime-rdel-1.tikz}{}{\input{./tikz/nf-prime-rdel-1.tikz}}
}$	
	\end{center}
	Thus,
	\begin{center}
		$\sem{
\InputIfFileExists{nf-rdel-2.tikz}{}{\input{./tikz/nf-rdel-2.tikz}}
} = \sem{
\InputIfFileExists{nf-prime-rdel-2.tikz}{}{\input{./tikz/nf-prime-rdel-2.tikz}}
}$	
	\end{center}
	which implies that $\sem{\thickboolcircuit{b}{p}{q}}=\sem{\thickboolcircuit{b'}{p}{q}}$ and, since these two circuits are Boolean, $b=b'$. 
	
	\noindent We now want to show that $d=d'$.
	Applying functoriality, equation~\eqref{eq:nf-equal} can be re-stated as follows: for any $a\in\Bool^p, a'\in\Bool^q$, and $x\in\R^m$
	$$\sem{b}(a'|a)\sem{d}(\cdot |a',a,x) = \sem{b'}(a'|a)\sem{d'}(\cdot |a'a,x)$$
	Moreover, recall that we have just shown that $b=b'$, so 
	\begin{equation}\label{eq:nf-equal-symbolic}
		\sem{b}(a'|a)\sem{d}(\cdot |a',a,x) = \sem{b}(a'|a)\sem{d'}(\cdot |a'a,x)
	\end{equation}
	There are two cases to consider:
	\begin{itemize}
		\item if $\sem{b}(a'|a)\neq 0$, we can simply divide both sides of~\eqref{eq:nf-equal-symbolic} by this number to derive $\sem{d}(\cdot |a',a,x) = \sem{d'}(\cdot |a'a,x)$;
		\item if $\sem{b}(a'|a)= 0$, the two sides of~\eqref{eq:nf-equal-symbolic} are also equal to zero; by convention (Definition~\ref{def:normal-form}), in this case, $\sem{d}(\cdot, a'| a,x)= \Normal{n}{0}{0}= \sem{d'}(\cdot, a'| a,x)$.
	\end{itemize}
	Thus, $\sem{d} =\sem{d'}$ and, by Theorem~\ref{thm:completeness-no-bool-output}, $d=d'$, as needed.
\end{proof}

\normalisationgeneral*
\begin{proof}
	We prove this by structural induction on circuits of type $\CGtype{p}{m}{q}{n}$. This is, we assume that for $\CGterm{c}{p}{m}{q}{n}$ an arbitrary circuit in NF in $\MixCirc$, composing with any generator $g\colon k\to l$ outputs in the monoidal signature $\Sigma_{\MixCirc}$, results in a circuit that is, again, in normal form.
	Similar to the proof of Theorem~\ref{thrm:completness-bprm-to-rn}, we first note that the possible composites fall into three categories:
	\begin{gather*}
		(i)\newline
		\noindent \:\: 
\InputIfFileExists{bprm-bqrn-shape01.tikz}{}{\input{./tikz/bprm-bqrn-shape01.tikz}}
 \qquad (ii) \:
\InputIfFileExists{bprm-bqrn-shape02.tikz}{}{\input{./tikz/bprm-bqrn-shape02.tikz}}
 \qquad (iii) \:
\InputIfFileExists{bprm-bqrn-shape03.tikz}{}{\input{./tikz/bprm-bqrn-shape03.tikz}}
 
	\end{gather*}
	For \textit{(i)}, we first tackle copiable generators in the Boolean fragment, \emph{i.e.}, we compose with some generator $g$ satisfying axiom \textax{C2}. We have,
	\begin{align*}
		
\InputIfFileExists{comp-bprm-bqrn-bool01.tikz}{}{\input{./tikz/comp-bprm-bqrn-bool01.tikz}}
 &\myeq{IH} 
\InputIfFileExists{comp-bprm-bqrn-bool02.tikz}{}{\input{./tikz/comp-bprm-bqrn-bool02.tikz}}
\\
		&\myeq{C2}
\InputIfFileExists{comp-bprm-bqrn-bool03.tikz}{}{\input{./tikz/comp-bprm-bqrn-bool03.tikz}}
\\
		&\myeq{Th.~\ref{thrm:completness-bprm-to-rn}}\;
\InputIfFileExists{comp-bprm-bqrn-bool04.tikz}{}{\input{./tikz/comp-bprm-bqrn-bool04.tikz}}

	\end{align*}
	where $	b' := \left(g\otimes \id_{\Boolobj^{p-l}}\right)\poi b$ and $d'\CGMeq \left(\id_{\Boolobj^q}\otimes \Andgate \otimes\id_{\Boolobj^{p-1}}\otimes \id_{\Realobj^m}\right)\poi d$ is the circuit in NF obtained by applying Theorem~\ref{thrm:completness-bprm-to-rn}. We have thus obtained a circuit in the shape given in Definition~\ref{def:normal-form}, but we still have to show that the side condition holds. Since $\sem{\cdot}$ is a monoidal functor, $\sem{b'}=\left(\sem{g}\otimes \id_{\Bool^{p-l}}\right)\poi\sem{b}$ 
	Assume for example that $g$ is $\Andgate$. Then, 
	$\sem{b'}(a'|a) = \sem{b}(a'|(a_0\land a_1)a_{[2:p]})=0$
	for $a'\in\Bool^q, a\in\Bool^p$, (and we use the notation $a_{[i:j]}\in\Bool^{j-i}$ to indicate the sub-array of $a$ between indices $i$ and $j$). Hence, $\sem{b'}(a'|a)=0$ implies $\sem{b}(a'|(a_0\land a_1)a_{[2:p]})=0$. Then, by the induction hypothesis, $\sem{d}(\cdot|a',(a_0\land a_1)a_{[2:p]},x)=\Normal{n}{0}{0}$. We want to show that $\sem{d'}(\cdot|a',a_{[0:2]}a_{[2:p]},x)=\Normal{n}{0}{0}$. Recall that
	\[
	d'\CGMeq \left(\id_{\Boolobj^q}\otimes \Andgate \otimes\id_{\Boolobj^{p-1}}\otimes \id_{\Realobj^m}\right)\poi d
	\]
	and thus, by monoidal functoriality again
	\begin{align*}
		\sem{d'}(\cdot|a',a_{[0:2]}a_{[2:p]},x) &= \big(\left(\id_{\Bool^q} \otimes\sem{\Andgate}\otimes\id_{\Bool^{p-1}}\otimes\id_{\Real^m} \right) \poi \sem{d} \big)(\cdot|a',a_{[0:2]}a_{[2:p]},x)\\
		& = \sem{d}(\cdot | a',(a_0\land a_1)a_{[2:p]},x) \\
		& = \Normal{n}{0}{0}
	\end{align*}
	where the last step is the induction hypothesis. 
	Thus $b'$ and $d'$ satisfy the normal form condition. A similar argument holds for all the other copiable generators.
	
	Let us now show how to normalise a diagram that is composed with $\Flip{p}$, the only non-copiable Boolean generator. We assume wlog that $p\in (0,1)$; otherwise $\Flip{p}$ is copiable and we can deal with it as in the previous cases. First, by the induction hypothesis, we have
	\begin{align*}
		
\InputIfFileExists{comp-bprm-bqrn-flip01.tikz}{}{\input{./tikz/comp-bprm-bqrn-flip01.tikz}}
 &\myeq{IH} 
\InputIfFileExists{comp-bprm-bqrn-flip02.tikz}{}{\input{./tikz/comp-bprm-bqrn-flip02.tikz}}

	\end{align*}
	The distribution $\sem{
\InputIfFileExists{flip-p-copy-bxid.tikz}{}{\input{./tikz/flip-p-copy-bxid.tikz}}
}$ can be disintegrated into the product of a marginal and a conditional; by Boolean completeness~\cite{probcirc}, we can mirror this decomposition diagrammatically---that is, we can find Boolean circuits $b'\colon\Boolobj^p\to\Boolobj^q$ and $b''\colon\Boolobj^p\Boolobj^q\to\Boolobj$ such that
	\[
\InputIfFileExists{flip-p-copy-bxid.tikz}{}{\input{./tikz/flip-p-copy-bxid.tikz}}
 \;\CGMeq\; 
\InputIfFileExists{flip-p-b-disintegrate.tikz}{}{\input{./tikz/flip-p-b-disintegrate.tikz}}
 \]
	with  $\sem{b'}=  \sem{\left(\Flip{p}\otimes\id_{\Bool^{p-1}}\right)\poi b}$.
	Thus,
	\begin{align*}
		 
\InputIfFileExists{comp-bprm-bqrn-flip02.tikz}{}{\input{./tikz/comp-bprm-bqrn-flip02.tikz}}

		= 
\InputIfFileExists{comp-bprm-bqrn-flip03.tikz}{}{\input{./tikz/comp-bprm-bqrn-flip03.tikz}}
 \myeq{B1} 
\InputIfFileExists{comp-bprm-bqrn-flip03-bis.tikz}{}{\input{./tikz/comp-bprm-bqrn-flip03-bis.tikz}}

	\end{align*}
	By Theorem~\ref{thrm:completness-bprm-to-rn}, the diagram in the dashed box above is normalisable, \emph{i.e.} equal to some diagram $d'$ in normal form and hence, the circuit $c$ is equal to one of the shape  required by Definition~\ref{def:normal-form}:
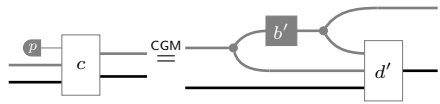
	\begin{center}
		$
\InputIfFileExists{comp-bprm-bqrn-flip01.tikz}{}{\input{./tikz/comp-bprm-bqrn-flip01.tikz}}
 \CGMeq 
\InputIfFileExists{comp-bprm-bqrn-flip04.tikz}{}{\input{./tikz/comp-bprm-bqrn-flip04.tikz}}
$
	\end{center}
	It remains to prove that the side condition of the normal form holds. We have that,
	\[
	\sem{b'}=\sem{\left(\Flip{p}\otimes\id_{\Bool^{p-1}}\right)\poi b}
	\]
	and, by monoidal functoriality,
	\[
	\sem{b'}=\left(\sem{\Flip{p}}\otimes\id_{\Boolobj^{p-1}}\right)\poi \sem{b}.
	\]
	From where
	\[
	\sem{b'}(a'|a)=p\cdot\sem{b}(a'|\mathsf{T}a)+(1-p)\sem{b}(a'|\mathsf{F}a)
	\]
	for $a\in\Bool^p,a'\in\Bool^q$ and $\mathsf{F},\mathsf{T}\in\Bool$. If $\sem{b'}(a'|a)=0$, since $p\in(0,1)$, it follows that $\sem{b}(a'|\mathsf{T}a)=0$ and $\sem{b}(a'|\mathsf{F}a)=0$. By the induction hypothesis, we then have that $\sem{d}(\cdot|a,\mathsf{T}a,x)=\Normal{n}{0}{0}=\sem{d}(\cdot|a,\mathsf{F}a,x)$ for all $x\in\Real^m$. Once again, by monoidal functoriality, we have
	\begin{align*}
	\sem{d'}(\cdot|a',a,x) &= \sem{
\InputIfFileExists{equal-d-prime.tikz}{}{\input{./tikz/equal-d-prime.tikz}}
}(\cdot|a',a,x) \\
	&= \sem{d}(\cdot|a',\mathsf{T}a,x)\sem{b''}(\mathsf{T}|a',a)+\sem{d}(\cdot|a',\mathsf{F}a,x)\sem{b''}(\mathsf{F}|a',a)
	\\
	&=\Normal{n}{0}{0}\sem{b''}(\mathsf{T}|a',a)+\Normal{n}{0}{0}\sem{b''}(\mathsf{F}|a',a) = \Normal{n}{0}{0}
	\end{align*}
	Thus complying with the side condition of the normal form.
	
\noindent We now consider case \textit{(ii)}, \emph{i.e.}, $g\in\Sigma_{\GaussCirc}$. We have, by the induction hypothesis,
	\begin{align*}
		
\InputIfFileExists{comp-bprm-bqrn-real01.tikz}{}{\input{./tikz/comp-bprm-bqrn-real01.tikz}}
 &\myeq{IH}
\InputIfFileExists{comp-bprm-bqrn-real02.tikz}{}{\input{./tikz/comp-bprm-bqrn-real02.tikz}}
\\
		&\myeq{Th.~\ref{thrm:completness-bprm-to-rn}}\;
\InputIfFileExists{comp-bprm-bqrn-real03.tikz}{}{\input{./tikz/comp-bprm-bqrn-real03.tikz}}

	\end{align*}
	where $d'$ is in normal form. 
	It remains to show that the side condition holds. If $\sem{b}(a'|a)=0$ for some $a\in\Bool^p,a'\in\Bool^q$, by the induction hypothesis, we have $\sem{d}(\cdot|a',a,x)=\Normal{n}{0}{0}$ for any $x\in\Real^m$. Moreover,
	$d' \CGMeq \big(\id_{\Boolobj^{q+p}}\otimes g\otimes\id_{\Realobj^{m-l}}\big)\poi d$ and, since $\sem{\cdot}$ is a monoidal functor, $\sem{d'} = \big(\id_{\Boolobj^{q+p}}\otimes \sem{g}\otimes\id_{\Realobj^{m-l}}\big)\poi \sem{d}$. Thus, 
	\begin{align*}
		\sem{d'}(V|a',a,x) &= \big(\left(\id_{\Bool^{q+p}}\otimes\sem{g}\otimes\id_{\Real^{m-l}}\right)\poi\sem{d}\big)(\cdot|a',a,x)\\
		&=\int_{y\in\Real^l}\Normal{n}{0}{0}(V)\sem{g}(dy|a',a,x)\\
		&=\int_{y\in\Real^l}\delta(V|0)\sem{g}(dy|a',a,x)\\
		&=\begin{cases}\int\sem{g}(dy|a',a,x)\text{ if }0\in V\\ 0 \text{ otherwise }\end{cases}\\
		& =\begin{cases}1 \text{ if }0\in V \text{ (since }\sem{g}\text{ is a probability distribution)}\\ 0 \text{ otherwise }\end{cases}\\
		&=\delta(V|0)\\
		&=\Normal{n}{0}{0}.
	\end{align*}
	Finally, for case \textit{(iii)}, we have the following:
	\begin{align*}
		
\InputIfFileExists{comp-bprm-bqrn-if01.tikz}{}{\input{./tikz/comp-bprm-bqrn-if01.tikz}}
 &\myeq{IH} 
\InputIfFileExists{comp-bprm-bqrn-if02.tikz}{}{\input{./tikz/comp-bprm-bqrn-if02.tikz}}
= 
\InputIfFileExists{comp-bprm-bqrn-if02-bis.tikz}{}{\input{./tikz/comp-bprm-bqrn-if02-bis.tikz}}
\\
		&\myeq{Th.~\ref{thrm:completness-bprm-to-rn}}\;
\InputIfFileExists{comp-bprm-bqrn-if03.tikz}{}{\input{./tikz/comp-bprm-bqrn-if03.tikz}}

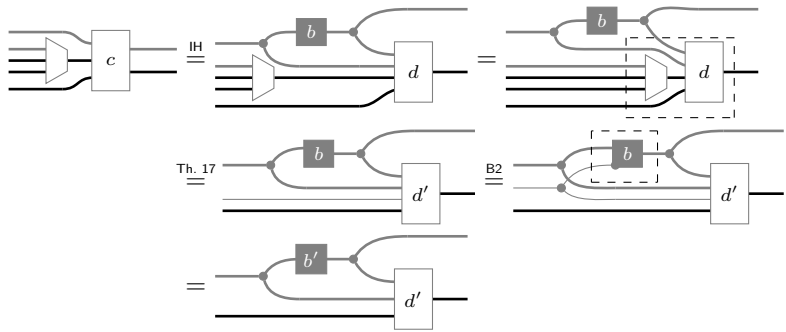
		\myeq{B2} 
\InputIfFileExists{comp-bprm-bqrn-if04.tikz}{}{\input{./tikz/comp-bprm-bqrn-if04.tikz}}
\\
		&= 
\InputIfFileExists{comp-bprm-bqrn-if05.tikz}{}{\input{./tikz/comp-bprm-bqrn-if05.tikz}}

	\end{align*}
	where $d'$ is obtained by normalising the circuit in the first dashed box above, and $b'$ is the second dashed box. 
	Finally, by a similar semantic reasoning as above, we can show that $\sem{d'}(\cdot|a',a,x)=\Normal{n}{0}{0}$ whenever $\sem{b'}(a'|a)=0$.
\end{proof}

\end{document}